%% file: main.tex
\documentclass[11pt]{article}

\def\bland{0} 

\input{packages}
\input{macros}

\def\anonymous{0} 
\titlespacing{\paragraph}{0pt}{1ex}{1em}

\begin{document}

\title{Interactive Proofs of Proximity for Model Evaluation}

\ifnum\anonymous=1
\author{}
\renewcommand{\nikolas}[1]{}
\renewcommand{\tamara}[1]{}
\renewcommand{\geoffroy}[1]{}
\else
\author[1]{Geoffroy Couteau}
\author[1]{Nikolas Melissaris}
\author[1,2]{Tamara Paris}

\affil[1]{Université Paris Cité, CNRS, IRIF}
\affil[2]{Université Paris Cité, McGill University}

\begingroup
\renewcommand\thefootnote{\arabic{footnote}}
\footnotetext[1]{\email{{couteau, nikolas}@irif.fr}}
\footnotetext[2]{\email{tamara.paris@mail.mcgill.ca}}
\endgroup
\fi

\date{}

\maketitle

\begin{abstract}
\input{abstract}
\end{abstract}

\section{Introduction}
\label{sec:intro}
\input{sections/intro}

\clearpage
\enlargethispage{\baselineskip}
\tableofcontents
\clearpage

\section{Preliminaries}
\label{sec:prelims}
\input{sections/prelims}

\section{Optimal dsIPP for Hamming Weight}
\label{sec:hamming_weight}
\input{sections/hamming_weight}

\section{dsIPPs for Distribution-Weighted Hamming Weight} 
\label{sec:distribution_weighted_hamming_weight}
\input{sections/unified-4-5}

\input{sections/weighted_hamming_weight}

\input{sections/corrupted_queries}

\section{Applications}
\label{sec:applications}
\input{sections/applications}

\section*{AI disclosure}
\addcontentsline{toc}{section}{AI disclosure}
\input{sections/disclosure}

\bibliographystyle{alpha}
\bibliography{refs/abbrev0, refs/crypto, refs/manual_refs}

\appendix

\section{The Banded Protocol for Evaluator Disagreement}
\label{app:banded-protocol}
\input{sections/appendix/banded_protocol}

\section{Optimality of the Banded Protocol}
\label{app:optimality-disagreement}
\input{sections/appendix/optimality_disagreement}

\end{document}

%% file: macros.tex
\newtheorem{theorem}{Theorem}[section]
\newtheorem{lemma}[theorem]{Lemma}
\newtheorem{proposition}[theorem]{Proposition}
\newtheorem{corollary}[theorem]{Corollary}

\theoremstyle{definition}
\newtheorem{definition}[theorem]{Definition}

\theoremstyle{remark}
\newtheorem{remark}[theorem]{Remark}

\newcommand{\eps}{\varepsilon}
\newcommand{\epsc}{\eps_{\mathrm{c}}}
\newcommand{\epsf}{\eps_{\mathrm{f}}}
\newcommand{\wt}{\mathrm{wt}}
\newcommand{\HW}{\mathsf{HW}}
\newcommand{\dwHW}{\mathsf{dwHW}}
\newcommand{\TV}{\mathsf{TV}}
\newcommand{\dsIPP}{\mathsf{dsIPP}}
\newcommand{\IPP}{\mathsf{IPP}}
\newcommand{\KL}{D_{\mathrm{KL}}}
\newcommand{\Ber}{\mathrm{Ber}}
\newcommand{\Prover}{\mathcal{P}}
\newcommand{\Verifier}{\mathcal{V}}
\newcommand{\Otilde}{\widetilde{O}}
\newcommand{\N}{\mathbb{N}}
\newcommand{\dist}{\Delta}
\newcommand{\PiSV}{\Pi_{\mathrm{SV}}}
\newcommand{\PiWM}{\Pi_{\mathrm{WM}}}
\definecolor{pastelgreen}{rgb}{0.47, 0.87, 0.47}
\definecolor{pastelviolet}{rgb}{0.8, 0.6, 0.79}
\definecolor{emailgreen}{RGB}{11,102,35}
\newcommand{\email}[1]{\href{mailto:#1}{\textcolor{emailgreen}{\texttt{\nolinkurl{#1}}}}}
\newcommand{\nikolas}[1]{%
  \todo[
    inline,
    nolist,
    bordercolor=pastelviolet!120,
    backgroundcolor=pastelgreen!30
  ]{%
    Nikolas:
    #1%
  }%
}
\newcommand{\tamara}[1]{\todo[inline,bordercolor=violet!120,backgroundcolor=pastelviolet!30]{Tamara: #1}}

\newcommand{\geoffroy}[1]{%
  \todo[
    inline,
    nolist,
    bordercolor=pastelviolet!120,
    backgroundcolor=yellow!30
  ]{%
    Geoffroy:
    #1%
  }%
}

\newlist{protodesc}{description}{1}
\setlist[protodesc]{font=\sffamily,wide}

\NewTColorBox{protobox}{ O{} m }{
    enhanced,
    title={\figurename~\thefigure: #2},
    attach boxed title to top left={
        xshift=3mm,
        yshift*=-3mm
    },
    colback=white,
    colframe=black!75,
    fonttitle=\bfseries,
    colbacktitle=black!10!white,
    coltitle=black,
    #1
}

\NewDocumentEnvironment{protofig}{s O{} m m}{
    \IfBooleanTF{#1}{
        \FloatBarrier
        \refstepcounter{figure}
        \label{#4}
        \begin{protobox}[breakable,#2]{#3}
    }{
        \begin{figure}[tbp]
        \refstepcounter{figure}
        \label{#4}
        \begin{protobox}[#2]{#3}
    }
}{
    \IfBooleanTF{#1}{
        \end{protobox}
        \FloatBarrier
    }{
        \end{protobox}
        \end{figure}
    }
}

%% file: abstract.tex
We study \emph{interactive proofs of proximity (IPP) for model evaluation}: a resource-limited verifier interacts with an untrusted prover, typically the model owner, to certify statistical properties of a model under an unknown input distribution. Our formulation is shaped by the constraints of practical evaluation: it separates sampling the input distribution from querying the model and evaluating its output, allowing their costs and access patterns to be treated independently; it distinguishes real audit data (black-box sampling) from generated data (chosen-randomness, or gray-box, access to the sampler); and it allows the prover and the verifier to score outputs with different evaluators, as happens when scores come from human or judge models. We focus on doubly-sublinear $\IPP$s, in which both the verifier and the designated honest prover use sublinear resources, and on (weighted) Hamming weight properties, which capture the expectation of a Boolean evaluation rule under an unknown distribution and, consequently, a broad range of model-evaluation statistics.

As a first step, we give a tolerant $\dsIPP$ for ordinary Hamming weight. For completeness and soundness radii $\epsc <\epsf$ and gap $g=\epsf-\epsc$, its constant-round instantiation uses $O(1/g)$ verifier queries and $O(1/g^2)$ honest-prover queries, improving the cubic dependence of Amir, Goldreich, and Rothblum [ITCS 2025]. We prove matching query lower bounds.

We then study distribution-weighted Hamming weight under several access models. Under black-box sampling, the verifier uses $\Theta(1/g^2)$ samples but only $O(1/g)$ evaluations, and we show that the quadratic sample complexity is necessary in the interior regime. With chosen-randomness access to a sampler, the problem reduces to ordinary Hamming weight, giving $O(1/g)$ verifier calls and evaluations. When the two parties' evaluators may disagree arbitrarily on a $\rho$-fraction of the distribution and by up to $\gamma$ elsewhere, we give protocols that remain doubly sublinear whenever $g$ exceeds twice the mean mismatch $\kappa = \rho + (1-\rho)\gamma$. Finally, we show how our technical results can improve the efficiency of model evaluation in natural motivating scenarios by shifting the bulk of the evaluation burden to the model owner while letting any number of auditors verify claims cheaply; we apply them to auditing criteria including accuracy, group fairness, calibration, harmlessness, usefulness, and average-case robustness.




%% file: sections/intro.tex
Interactive Proofs of Proximity ($\IPP$s), first introduced by \cite{fast_approximate_probabilistically_2004} and then further conceptualized by \cite{STOC:RotVadWig13}, enable a verifier to efficiently distinguish inputs in a language from inputs far from the language by interacting with an untrusted, but more powerful, prover. This extends the classical \emph{property testing} framework \cite{10.1145/285055.285060,10.1137/S0097539793255151}, by allowing the verifier to delegate part of the computation and queries to the untrusted prover. More recently, \cite{ITCS:AmiGolRot25} initiated the study of doubly-sublinear Interactive Proofs of Proximity ($\dsIPP$s) in which the prover query complexity is also sublinear and not much larger than testing.

$\IPP$s naturally apply to the field of machine learning, where querying the input domain can be costly and where the empirical nature of the evaluations is particularly well suited to probabilistic proximity guarantees. This has already been highlighted in several previous works, including \cite{ITCS:GRSY21} who introduced the notion of \emph{PAC verification}, which leverages $\IPP$s to verify if a hypothesis (typically a machine learning model) is a good hypothesis among its class, and \cite{FOCS:HerRot24}, who generalized the notion of interactive proof for distribution properties to a large class of properties, including non-label-invariant properties.

Our work builds on these previous efforts by studying \textit{interactive proofs of proximity for model evaluation}, which we conceptualize as tolerant $\dsIPP$s that, given query access to a function $f$ and sample access to a distribution $D$, allow one to verify that the distribution of $f(X)$ for $X \sim D$ is close to a certain distribution property. We propose several protocols across different access models for the distribution-weighted Hamming Weight problem ($\dwHW$), from which many machine learning evaluation measures can be derived.
We do so by first constructing a $\dsIPP$ for the standard Hamming Weight problem that is optimal, thereby improving on the upper bounds of \cite{ITCS:AmiGolRot25} (see Sec.~\ref{sec:hamming_weight}). Then, building on our previous result, we study $\dsIPP$s for $\dwHW$ under several access models (see Sec.~\ref{sec:distribution_weighted_hamming_weight}). We first consider black-box sampling and chosen-randomness access to a sampler, and then focus on the setting in which the prover and verifier query two different evaluation functions. Finally, we explain how to build derived $\dsIPP$ protocols for numerous widely used machine learning evaluation measures, and which challenges in machine learning inform the access models studied (see Sec.~\ref{sec:applications}). 

\subsection{Our Focus: dsIPPs for Model Evaluation}

\paragraph{Our Setting.} Let $D$ be a distribution over an input space $\mathcal{X}$. The exact description of $D$ is unknown to the verifier and prover, however, they can both sample from the distribution. The prover and verifier can query a model $M:\mathcal{X} \to \mathcal{R}$, which can represent either a supervised or an unsupervised machine learning model. They can also query a labeling function $\phi: \mathcal{X} \times \mathcal{R} \to \mathcal{Y}$, with $\mathcal{R}$ the output space, and $\mathcal{Y}$ the space of evaluation labels. The labeling function serves as an intermediate step in the machine learning evaluation procedure and may encode different kinds of evaluation information: for example, it can indicate how close the output is to a ground truth label, it can associate a safety score to it, or it can categorize input-output pairs in unordered categories. Even if we represent $\phi$ as a function, $\phi$ may represent non-computational processes, such as human labeling and physical-world measurements in the context of human-computer interaction. On the other hand, the labeling function $\phi$ may also directly return the output if no intermediate labeling is needed for the evaluation. Finally, the evaluation function $f_{M,\phi}: \mathcal{X} \to \mathcal{Y}$ is defined by $f_{M,\phi}(x)=\phi(x, M(x))$. In the rest of the paper, we use $f_{M,\phi}$ and $f$ interchangeably.

We conceptualize the \emph{interactive proofs of proximity for model evaluation} as tolerant $\dsIPP$s that allow the prover to convince the verifier that the distribution $D_f$, defined by $D_f(y)=\Pr_{X\sim D}[f(X)=y], \forall y\in \mathcal{Y}$, is close to a given property $\Pi$. We focus on \emph{doubly-sublinear} proofs because sampling $D$ and querying\footnote{Although $f$ may involve non-computational processes through $\phi$, we refer to executing $f$ as ``querying'' throughout the paper, in keeping with the standard terminology used in the $\IPP$ literature. However, we acknowledge that the process represented by $f$ may be very different from what the term ``querying'' suggests.} $f$ are often also costly for the prover, which is typically the model owner in the machine learning context. Moreover, because these operations can be costly, the prover itself often only has an approximate knowledge of whether $D_f$ satisfies $\Pi$. We therefore focus on \emph{tolerant} proofs, in which completeness holds for all distributions that are $\epsc$-close to $\Pi$ for a proximity parameter $\epsc \geq 0$.

The flexibility of $\phi$ allows us to capture a wide range of evaluation statistics, from standard accuracy measures in supervised learning to quantities obtained through real-world experiments, such as how well a given population understands text generated by the model. The nature of the function $\phi$ affects not only the cost of the function $f_{M,\phi}$, but also its reliability. This leads us to discuss the access models in more detail.

\paragraph{Two Access Models: Sampling and Querying.}

Compared to previous work on $\IPP$s for distribution properties, the main novelty of our approach to $\dsIPP$s lies in the access models. If we consider that the parties have sampling access without the ability to control the seed and that the cost of a query is negligible or included in the sampling cost, then the resulting proof model is analogous to the framework of interactive proofs for general distribution properties introduced in \cite{FOCS:HerRot24}. However, in the context where the prover also uses sublinear samples, \cite{FOCS:HerRot25} demonstrate a negative result for many natural distribution properties, including testing whether the distribution is uniform and specifying its $L_k$ norm: non-trivial $\dsIPP$s for these properties are impossible under this access model. In this work, we pursue an application-specific perspective to motivate new access models that have not been previously considered in the $\IPP$ literature, both to account separately for the cost of querying the model and to consider pragmatic workarounds for the negative results of \cite{FOCS:HerRot25}:
\begin{quote}
    \textit{Which access models reflect the reality of machine learning evaluation? Under which of these access models is it possible to construct a non-trivial $\dsIPP$ for $\dwHW$?}
\end{quote}
In machine learning, collecting a sample on which to evaluate a model (i.e., sampling $D$), and querying the model and labeling its input-output (i.e., querying $f_{M,\phi}$) are distinct actions that can incur very different costs depending on the context. For example, sampling synthetic data is less costly than sampling ``in-the-wild'' data, and evaluating the output is less costly when it only involves computations than when it requires manual labeling by experts. Therefore, a $\dsIPP$ that significantly reduces the query complexity while maintaining a sample complexity close to that of a tester is still useful in settings where querying is significantly more costly than sampling.

\paragraph{Sample Access Models: Black-Box vs. Gray-Box.} Beyond cost, the sample access model itself may also differ from one application context to another. We study two of these possible access models: \emph{black-box} and \emph{gray-box} sampling.

\emph{Black-box sampling} refers to collecting a fresh independent data point $X \sim D$. The caller cannot choose or inspect the randomness used to generate $X$. This access model typically represents the collection of data ``in-the-wild'', such as recording user-chatbot conversations or collecting student records from universities. Of course, real-world data collection conditions are imperfect, making it very difficult to guarantee fresh, independent samples from $D$. In particular, cases have been reported of model owners overfitting to the biased distribution used during the evaluation, ending up with non-generalizable results\footnote{In practice, overfitting to the evaluation dataset can occur even in the absence of a biased evaluation distribution: the model owner only needs to know the dataset itself. However, in this work, we assume that the data collection is performed after the training of the model $M$, and that $M$ is not modified at any point during the process. Thus, the model owner does not have prior knowledge of the data points that will be sampled by the verifier. Under these assumptions, constructing the evaluation dataset from i.i.d. samples from $D$ is sufficient to prevent the model from overfitting to the evaluation samples.}~\cite{leaderboard_illusion_2025}. Nevertheless, this abstraction can still be useful in many practical settings where the data collection is performed directly at the source of the data. For example, if all uses of a model are publicly observable and $D$ represents the distribution of these uses, then it is reasonable to abstract the data collection process using the black-box sampling model.

\emph{Gray-box sampling} refers to using a deterministic sampling oracle $G\colon\{0,1\}^{R}\longrightarrow \mathcal{X}$ such that $G(U_R)\sim D$, where $U_R$ is uniform over $\{0,1\}^{R}$. This access model applies naturally to the generation of synthetic data. The verifier and the prover can use the same algorithm to generate the data with a shared random tape of their choice. The random tape can be selected through a trusted third party or through a \emph{nothing-up-my-sleeve} procedure, which leverages public and hard-to-guess values (e.g., lottery outcome). On the other hand, the purpose of an evaluation is rarely to assess the model on the synthetic distribution itself. We discuss this limitation in more detail in Section~\ref{sec:applications}.

\paragraph{Query Access Models: Evaluator Disagreement.} When the process abstracted by $\phi$ is not entirely
computational, the verifier and the prover share instructions on how to execute the process $\phi$. However, even if these instructions are well-detailed, the prover and verifier may not always agree on the result of the execution. For example, in contexts requiring manual labeling, the verifier and the prover may assign different labels to outputs in ambiguous cases. This is particularly common in evaluation domains in which the existence of a ground truth is debated, including safety, ethics and cultural appropriateness evaluations. 

To address this concern, we consider query access models in which the verifier and the prover query two different functions $f_\Verifier$ and $f_\Prover$ that have bounded $(\rho,\gamma)$-mismatch with respect to $D$, meaning there is a subset $B\subset\mathcal{X}$ such that $D(B)\leq \rho$ and $|f_\Prover(x)-f_\Verifier(x)|\leq \gamma$ for every $x \notin B$. In other words, except for the small proportion of inputs in $B$, the outputs of $f_\Prover$ and $f_\Verifier$ are close. However, no restriction is imposed on the outputs in $B$. This captures the cases where the verifier and the prover agree on most samples but, in rare cases, strongly disagree on the resulting output of $\phi$.

\subsection{Overview of Results}

We describe several protocols which can be used to audit properties of machine learning models in frameworks that capture the ability of auditors to interact with a powerful but untrusted model owner. All our protocols achieve quadratic gains in their resource access for the verifier compared to the resource-cost of auditing the same property without help from the model owner (matching the typical benefits of doubly-sublinear interactive proofs of proximity) in a framework designed to abstract out the core features of real-world model auditing; we complement our study with lower bounds capturing the limits of this framework.

\subsubsection{Hamming Weight} 

We first start with an intermediate result: a construction of a $\dsIPP$ for the standard Hamming Weight problem ($\HW$) which is optimal. Beyond providing an answer to the open problem raised by \cite{ITCS:AmiGolRot25}, this intermediate result extends nicely to the distribution-weighted Hamming Weight, which is our ultimate goal.

We note first that a very simple spot-checking protocol already achieves the optimal query complexities: the verifier sends a short description of $m=O(1/g^2)$ pairwise-independent positions of $x$, the prover returns the $m$ corresponding bits, and the verifier checks that their mean is close to the claimed relative weight and spot-checks $O(1/g)$ of them against $x$. This one-round protocol has verifier query complexity $O(1/g)$ and honest-prover query complexity $O(1/g^2)$, but its prover-to-verifier communication is $\Theta(1/g^2)$. Our protocol achieves the same query complexities with $\Otilde(1/g)$ communication. This saving is modest for $\HW$ itself, but it is what makes our gray-box protocol for distribution-weighted Hamming weight (see below) efficient in all resources.

Our protocol starts in the same way (resp. with $m=O(1/\eps^2)$ in the non-tolerant setting), but the prover only sends the Hamming weight of the sampled virtual string $y$ composed of these $m$ positions. After verifying that the sent weight divided by $m$ is close to the claimed one divided by $n$, the verifier and prover run as a subprotocol any $\IPP$ protocol for $\HW$, with the substring $y$ as input. This leads to the following theorem:

\begin{theorem}[$\dsIPP$ for $\HW$, informal version of Theorem~\ref{thm:dsIPP_hw}] Let $\epsc$ and $\epsf$ be the proximity parameters such that $\epsf > \epsc \geq 0$, let $g:=\epsf-\epsc$ be the tolerance gap, and let $n$ be the input length for the Hamming Weight problem ($\HW$). Let $\Pi_{\mathrm{HW}}$ be an $r$-round $\IPP$ for $\HW$ with verifier's query complexity $\mathrm{Q_{\Pi_{\mathrm{HW}}}}(n)$ and communication complexity $\mathrm{C_{\Pi_{\mathrm{HW}}}}(n)$. Given $m=O(1/g^2)$, there exists a tolerant $(r+1)$-round $\IPP$ with verifier's query complexity in $\mathrm{Q_{\Pi_{\mathrm{HW}}}}(m)$, prover's query complexity in $m=O(1/g^2)$ and communication complexity in $O(\log n) + \mathrm{C_{\Pi_{\mathrm{HW}}}}(m)$. This $\IPP$ is therefore doubly sublinear when $1/g^2=o(n)$.
\end{theorem}

\noindent In particular, if $\Pi_{\mathrm{HW}}$ is the $r$-round $\HW$ protocol in \cite{ITCS:AmiGolRot25}, run on the sampled string of length $m=O(1/g^2)$, we obtain a tolerant $(r+1)$-round $\dsIPP$ such that the verifier's query complexity is $O(1/g)$, the prover's query complexity is $O(1/g^2)$ and the communication complexity is $O((1/g^2)^{1/r} \cdot r \cdot \log(r/g)/g + \log n)$. We also describe a one-round protocol (see Figure~\ref{fig:one_round_dsIPP_hamming_weight}), which combines the first message of our general protocol with the first message of the one-round protocol of \cite{STOC:RotVadWig13}, instead of just instantiating it as a subprotocol. This one-round protocol achieves verifier's query complexity in $\Otilde(1/g^{4/3})$, prover's query complexity in $O(1/g^2)$ and communication complexity $O(\log n) + \Otilde(1/g^{4/3})$.

We complement our results with matching query lower bounds: let $Q_\Verifier$ and $Q_\Prover$ be the verifier's and prover's query complexity respectively, then every tolerant $\IPP$ for the Hamming weight problem satisfies $Q_\Verifier=\Omega(1/g)$ and $Q_\Verifier+Q_\Prover=\Omega(1/g^2)$ (see Proposition~\ref{prop:hw-query-lower-bounds}). In addition, we prove that there is no $\dsIPP$ with perfect completeness (see Proposition~\ref{prop:no-perfect}).

\subsubsection{Distribution-Weighted Hamming Weight}

Let $D$ be a distribution over a domain $\mathcal{X}$ and $f\colon\mathcal{X}\to\{0,1\}$ an evaluation function. We want to verify a claim $\sigma$ about \(\wt_D(f):=\mathbb{E}_{X\sim D}\left[f(X)\right]\). Equivalently, if $D_f$ denotes the distribution of $f(X)$ for $X\sim D$, then $\wt_D(f)$ is the mean of $D_f$. Unlike in the standard Hamming Weight problem, the distribution $D$ is unknown. Therefore, we distinguish the cost of sampling from $D$ from the cost of querying $f$.

The access available to the verifier determines how closely this problem reduces to ordinary Hamming Weight. We first consider ordinary black-box sampling, then chosen-randomness access to the sampler, and finally the case where the prover and verifier evaluate sampled points differently.

\paragraph{Black-Box Sampling.}
We first consider the setting where the verifier can only obtain fresh independent samples from $D$. For tolerance gap
$g=\epsf-\epsc$, our protocol uses $O(1/g^2)$ samples from $D$, but the verifier evaluates $f$ on only $O(1/g)$ of them. The honest prover evaluates $f$ on at most $O(1/g^2)$ sampled points.

The protocol first draws an empirical sample $X_1,\ldots,X_m$ of size $m=O(1/g^2)$ and defines the virtual string 
\( Y_i:=f(X_i)\). The prover computes the Hamming weight of $Y$, while the verifier checks the prover's claim using the Hamming Weight protocol from Section~\ref{sec:hamming_weight}. Concentration of the empirical mean ensures that the relative Hamming weight of $Y$ is close to $\wt_D(f)$. Thus the interaction reduces the number of evaluations of $f$ performed by the verifier, even though the verifier still has to draw the full sample.

Here too, the spot-checking baseline achieves the same query complexities in a single round. Since the verifier must send the whole sample anyway, our communication saving is minor in this setting, and we present this protocol mainly as a stepping stone towards the gray-box setting.

The quadratic sample complexity cannot in general be improved. In the interior regime, where the relevant means remain bounded away from $0$ and $1$, we prove that every such interactive proof requires $\Omega(1/g^2)$ black-box samples, independently of its number of rounds, communication, or number of queries to $f$. Thus interaction can reduce the number of model evaluations performed by the verifier, but not the number of fresh samples needed to identify the underlying mean.

\paragraph{Gray-Box Sampling.}
Things are different when the parties have chosen-randomness access to a sampler. Suppose that 
\(G\colon \{0,1\}^R\to\mathcal{X}\) satisfies \(G(U_R)\sim D\), and that the parties may evaluate $G$ on a random tape of their choice. Defining \(z(r) := f\left(G(r)\right)\) reduces $\dwHW$ exactly to the ordinary Hamming Weight problem on the length-$2^R$ string $z$. Applying our result from Section~\ref{sec:hamming_weight} therefore gives a logarithmic-round protocol in which the verifier makes $O(1/g)$ chosen-randomness calls to $G$ and $O(1/g)$ queries to $f$, while the honest prover makes $O(1/g^2)$ calls and queries. The communication complexity is $O(R)+\Otilde(1/g)$. Since the sample is now described by a short seed, the $\Theta(1/g^2)$ prover message of the spot-checking baseline would dominate this cost.

This shows the distinction between the two sampling models because with black-box access, $\Theta(1/g^2)$ verifier samples are necessary in the interior regime, while chosen-randomness access allows the verifier to reduce its use of the sampler to $O(1/g)$.

\paragraph{Evaluator Disagreement.}
Finally, we consider settings in which the prover and verifier may not obtain exactly the same evaluation on a sampled point. We model their evaluations by two fixed functions $f_{\Prover}$ and $f_{\Verifier}$. We say that they have bounded $(\rho,\gamma)$-mismatch if their scores may differ arbitrarily on a set of probability at most $\rho$, but differ by at most $\gamma$ everywhere else. Writing \(\kappa := \rho+(1-\rho)\gamma\), their expected disagreement is
at most $\kappa$, and therefore the difference between their means is also at most $\kappa$.

The main difficulty is that the verifier can no longer directly check the values used by the prover. We therefore use a binary-splitting protocol to reduce a claimed sample sum to a randomly selected leaf. At the leaf, instead of checking equality between the prover's and verifier's evaluations, the verifier performs a randomized comparison whose rejection probability is exactly their absolute difference. As a result, an honest prover is charged only for the disagreement already present between the two evaluation functions, whereas a false claim about the mean still forces a noticeable rejection probability.

This disagreement effectively reduces the tolerance gap from $g$ to \(\eta:=g-2\kappa\). When $\eta>0$, we obtain a protocol using $O(1/\eta^2)$ honest-prover evaluations and $O(g/\eta^2)$ verifier evaluations for constant error. In particular, if $\kappa\leq cg$ for any fixed $c<1/2$, we recover the same $O(1/g^2)$ honest-prover and $O(1/g)$ verifier evaluation complexities as in the common-evaluator setting. The protocol also admits a gray-box implementation with the corresponding bounds on chosen-randomness calls.

\subsection{Applications}

\paragraph{Auditing machine learning metrics.} Our focus on verifying a distribution-weighted Hamming weight may look like a modest goal, but it is intentional: a wide variety of useful model properties can easily be \emph{reduced} to a small number of distribution-weighted Hamming weight tests. Therefore, a conceptually simple protocol provides a unified framework to audit multiple distinct properties.

In Section~\ref{sec:audit-criteria}, we formalize this observation as follows: we consider an audit record composed of an input $X$, the model output $M(X)$, and possibly some metadata (e.g., a ground truth label or a protected attribute), together with a Boolean rule $\phi$ that represents a decision about the record (e.g., ``the prediction is wrong''). The probability that a record sampled from $D$ satisfies $\phi$ is exactly $\wt_D(f)$ for $f(x)=\phi(x,M(x))$, so any property which can be framed as the rate of records accepted by a decision rule $\phi$ reduces to one $\dwHW$ claim about the ``model'' $f = \phi(\cdot,M(\cdot))$. Furthermore, properties that can be framed as \emph{conditional rates} w.r.t. a decision rule $\phi$ are ratios of two such weights, and comparing two conditional rates, which is what most group fairness criteria do, amounts to certifying four weights and then checking a simple inequality between them.

This is enough to cover a large fraction of the model evaluation tasks that have been considered in the literature. Accuracy (does the model often produce correct answers to the questions?) is one weight; demographic parity and equal opportunity~\cite{3157382.3157469}, two standard fairness metrics, compare two conditional rates, hence use four weights, and equalized odds (another common fairness metric) uses eight~\cite{3157382.3157469}; the rate of harmful (resp. useful) outputs is one weight for a rule that is typically evaluated by humans (does the model often produce answers labeled as potentially harmful/useful?), which is precisely the situation motivating our evaluator disagreement model. Bounded scores, such as an average rating in $\{0,1/K,\ldots,1\}$, are handled by a simple trick: the average of a score over $D$ is the weight of a Boolean rule over $D$ combined with a uniform threshold in $[K]$, so it is again a $\dwHW$ claim on a slightly larger domain. Calibration at a given score value, or on a public bin,
reduces to two weights and a comparison, and average-case robustness to a perturbation is the weight of the rule ``the predictions on $X$ and on its perturbation $\widetilde X$ differ''. In all these cases, the tolerance of the audit has to be adapted to the probability of the conditioning events: certifying something about a rare group is more expensive, unless the sampler can directly produce records from that group.

\paragraph{Application of the access models.} Our motivating scenario is the following (Section~\ref{sec:generated-audits}):
a model owner wants to certify a public statistical claim about a fixed model, and producing one evaluated audit record is expensive, since it may involve generating an input, running the model, and having a human judge the output. Our protocols let the model owner perform the $O(1/g^2)$ evaluations supporting the claim, while the auditor only checks $O(1/g)$ of them. How exactly this plays out depends on the access model.

When the audit data is synthetic, i.e., produced by a public generator $G$, the gray-box protocol applies directly: the auditor makes $O(1/g)$ calls to $G$ and $O(1/g)$ evaluations, and the audit set is described by a short seed rather than transmitted. Better, the owner's work can be done once and for all (Section~\ref{sec:public-audit-set}). If the random seed for the audit set is derived from a public and unbiasable value fixed after the model (e.g., a lottery outcome, or random string periodically produced by an already-deployed cryptographic random beacon protocol, such as the League of Entropy~\cite{LoE20,gailly2023tlock}) then the audit set is public and reproducible: the owner evaluates the model on it a single time, at a cost of $O(\log(k/\delta)/g^2)$ evaluations, and any number of auditors can then verify claims about any of $k$ fixed criteria with $O(1/g)$ evaluations each. When the audit data comes from a real population and can only be sampled, the situation is less favorable (Section~\ref{sec:audit-access-comparison}): the auditor must collect $\Theta(1/g^2)$ fresh samples itself to send to the owner, and amortizing the work does not seem feasible. In Section~\ref{sec:audit-certificate-scope}, we discuss some known limitations of the use of synthetic data as a proxy to measure a model behavior on a real population~\cite{3618408.3619856}: the validity of the proxy is only guaranteed when the total variation distance between the two distributions (synthetic and in-the-wild) is small, which our protocols do not certify, and conditional criteria on small groups are particularly sensitive to this distance.

\subsection{Discussion and Perspectives}
\label{sec:discussion}

This paper sits between two lines of work: the theory of interactive proofs of proximity fixes a proof model and studies which properties admit efficient protocols in it; the literature on auditing machine learning models asks what should be measured, on which data, and by whom. Our
contribution is neither aimed at advancing the theory of interactive proofs of proximity nor at proposing an end-to-end usable auditing system. Rather, it is a framework that connects the two, together with concrete protocols and lower bounds that characterize the limits of what can be achieved. For a contribution of this kind, the conceptual message matters as much as the theorems, so we state it explicitly.

\paragraph{Access models are chosen for the application, not for their novelty.}
Some of our models are, formally, restatements of existing ones. The clearest case is gray-box sampling: write the input as the pair $(G,f)$ (the table of the generator $G$ on its $2^R$ random tapes followed by the table of the evaluation function $f$). A chosen-randomness call to $G$ and a query to $f$ are then ordinary queries to coordinates of this string, so a gray-box protocol is, formally, a standard $\dsIPP$ for an unusual property: the certified quantity, $2^{-R}\sum_r f(G(r))$, is an average over the $f$-part of the input with weights determined by the $G$-part. One chunk of the input lists which values matter, the other chunk supplies them. Viewed this way, we are merely studying a traditional dsIPP for an unconventional property. What motivates our reframing is that it matches much more closely the intended application (auditing a Hamming-like property of a model weighted by a synthetic data distribution). Furthermore, auditing a model on synthetic data means running a public generator on chosen randomness and then evaluating the model on what it produces: these are two operations with different costs, and our model keeps them separate. The same holds, to a lesser extent, for black-box sampling, which is close to the model of interactive proofs for distribution properties~\cite{FOCS:HerRot24}. What our formulation adds is the separation between sampling the population and evaluating the model, which is exactly where the asymmetry between prover and verifier lies.

\paragraph{One property, several proof models.}
Work on proofs of proximity usually fixes a proof model and explores the properties it can handle. We proceed in the opposite direction. Section~\ref{sec:audit-criteria} shows that a single property, the distribution-weighted Hamming weight, covers a large part of what one wants to certify about a model: accuracy, group fairness criteria,
calibration, harmlessness and usefulness rates, and average-case robustness are all averages of a Boolean or bounded evaluation rule over the input distribution, or comparisons between a few such averages. We therefore fix this one property and vary the proof model along the dimensions that distinguish auditing scenarios: whether the audit data is collected from a real population or generated (black-box
versus gray-box sampling), and whether the two parties evaluate outputs identically or through different evaluators (the common-evaluator and evaluator-disagreement settings). Each variation changes the achievable complexity, and in each case we prove lower bounds showing that the protocols we obtain are essentially the best possible.

\paragraph{A quantified trade-off for synthetic audit data.}
Synthetic data is known to be a weaker basis for evaluation than data collected from the deployed population: generators may misrepresent that population, in particular in low-density
regions~\cite{3618408.3619856}, and a certificate obtained on generated data transfers to the real population only under an additional fidelity assumption (Section~\ref{sec:audit-certificate-scope}). Our results quantify the other side of this trade-off. With data from a real population, the verifier needs $\Theta(1/g^2)$ fresh samples, and no
interaction with a prover can reduce this number
(Theorem~\ref{thm:black-box-sample-lower-bound}). With a generator, the verifier needs only $O(1/g)$ calls to it and $\Otilde(1/g)$ bits of communication beyond the description of a seed (Theorem~\ref{thm:gray-box-whw}), while its number of model evaluations is $O(1/g)$ in both cases. Choosing synthetic data for an audit thus buys a quadratically cheaper verifier (in settings where a prover-assisted audit is feasible) at the price of a weaker guarantee about deployment. We believe that conveying this trade-off precisely is useful to make an informed decision about how to best audit a model in a given setting.

\paragraph{Assuming a powerful prover is a well-motivated model for auditing.} Interactive proofs of proximity assume a powerful prover and explore how it allows going beyond the efficiency limitations of property testing. A central thesis of our work is that this assumption is not just a technical convenience: it is reasonable and well-motivated in the setting of model evaluation. This rests on three observations: first, the party that benefits from a favorable audit is the model owner, who may need a certificate (of successfully passing an established auditing process) to deploy a model or to sell access to it. It is reasonable that this party bears the cost of producing the evidence, and soundness guarantees that the auditor need not trust it; one can imagine the owner initiating the audit and presenting the resulting certificate to regulators or customers. Second, the owner typically has far more computational and evaluation resources than any single auditor (e.g., in the case where it is a major AI company). Third, under certain (acceptable) assumptions, the owner's work can be done once and reused across any number of audits by independent auditors. Concretely, if the audit uses synthetic data and if the random tapes of the audit set are
derived from a public random beacon, such as a lottery draw or any publicly-verifiable random value fixed after the model and the audit criteria (many practical instantiations of such a trusted setup already exist), then the owner evaluates the model on the $O(1/g^2)$ generated inputs a single time, and any number of auditors can subsequently verify any number of
claims about it with $O(1/g)$ evaluations each
(Section~\ref{sec:public-audit-set}). Independent audits would each pay the full evaluation cost; the beacon, a mild assumption, factors this cost across auditors and across criteria. For audits using in-the-wild data, the prover work can be factored similarly, but only under the assumption that a trusted auditor gathers the audit set and makes it publicly available, which is a much stronger (and much less reasonable) trust assumption. We stress that this result is not meant to advocate for the systematic use of synthetic inputs in model audits (whether their use is acceptable or not is heavily context-dependent) but rather to characterize precisely how much we can leverage interactions with the model owner once the auditing methodology is defined. 

\paragraph{Perspectives.} We end this discussion with some perspectives for future works. First, we show that model auditing can be made much more efficient in a gray-box setting, but this setting implicitly assumes that the synthetic data accurately mimics the distribution of real-world data. We believe that this motivates further work on the statistical problem of establishing the fidelity of an input-generator to the ``natural'' input distribution for a deployed population (see Section~\ref{sec:audit-certificate-scope}). Second, our soundness guarantees are information-theoretic against a prover who knows the instance, but the model itself is \emph{assumed fixed before the audit}: an owner who could change the model after learning the audit set falls outside the framework. The question of mitigating this assumption has been considered at length in the model auditing literature, using either cryptographic proofs to reduce the freedom of a cheating prover, or by studying the resilience of audits to adaptive model manipulations~\cite{yan2022active,godinot2024under}. Connecting these questions with our approach is one of the directions we consider most important.

\subsection{Further Related Work} 

\paragraph{Interactive Proofs for distribution properties.}
Proof systems for distribution properties were introduced by Chiesa and Gur \cite{ITCS:ChiGur18} and subsequently developed by Herman and Rothblum \cite{STOC:HerRot22, FOCS:HerRot23, FOCS:HerRot24, FOCS:HerRot25}. In this model, the object being verified is an unknown distribution $D$. Given a property $\Pi$ of distributions, the verifier must distinguish distributions $\epsc$-close to $\Pi$ (in total variation distance) from those $\epsf$-far from $\Pi$ using black-box sample access to $D$ and interaction with an untrusted prover. While only label-invariant properties were first studied in \cite{ITCS:ChiGur18, STOC:HerRot22, FOCS:HerRot23}, the proof model was then extended to non-label-invariant properties in \cite{FOCS:HerRot24, FOCS:HerRot25}, thereby making this line of work more relevant to machine learning applications.

As we described above, their proof model differs from ours mainly in the access models. Our protocol also aims to verify the proximity to a general distribution property, with the pushforward distribution $D_f$ as the input distribution. Alternatively, for a fixed public function $f$, our proof model can be viewed as testing a particular property of $D$. However, we do not provide sample access to $D_f$ but a separate access to $D$ and to $f$ with different and non-negligible costs.
The two models nevertheless intersect on the hard instances used to derive the sample lower bound in the black-box model (see Theorem~\ref{thm:black-box-sample-lower-bound}). 


\paragraph{Distribution-free IPPs.} Distribution-free IPPs \cite{aaronson_et_al:LIPIcs.CCC.2024.24} are based on a similar access model to ours: the verifier may query a function $f$ and sample from an unknown distribution $D$. However, the property tested has a different form. For a property $\Pi$ of functions, the verifier distinguishes $f\in\Pi$ from $\Delta(f,\Pi):=\inf_{h\in\Pi}\Pr_{X\sim D}[f(X)\neq h(X)]>\eps.$ Unlike our proof model, membership in $\Pi$ is independent of $D$; the distribution determines only how distance from $\Pi$ is measured.

In our model, changing $D$ while keeping $f$ fixed can turn a completeness instance into a soundness instance. This is exactly what we leverage in the proof of Theorem~\ref{thm:black-box-sample-lower-bound} to derive the black-box sampling lower bound. By contrast, in the distribution-free model, if a fixed function $f$ is a completeness instance, then $f\in\Pi$ and remains a completeness instance for every choice of $D$. Hence our sampling lower bound does not apply to distribution-free properties. 


\paragraph{PAC verification.} Our proof model also resembles the PAC-verification framework of Goldwasser, Rothblum, Shafer, and Yehudayoff~\cite{ITCS:GRSY21} (perhaps the closest work to ours in terms of goals), which is also designed for machine learning applications. In one of their settings, the verifier receives random labeled examples while the honest prover has the stronger ability to make membership queries. Their interactive
protocol delegates such queries to the prover while arranging that some of the answers can be checked against examples already known to the verifier. Our protocol uses the same general principle where the prover performs the larger set of evaluations, while the verifier uses independently obtained data to bind the prover's claims. However, the verification goals are different: PAC verification concerns a hypothesis class $\mathcal{H}$ and asks for a
hypothesis $h$ satisfying a guarantee of the form
\[
    L_D(h)
    \leq
    \alpha\inf_{h'\in\mathcal{H}}L_D(h')
    +
    \eps,
\]
where $L_D$ denotes the loss function. By contrast, our protocols verify the proximity to any property claimed, including properties that indicate evaluation results far from the optimal results attainable by the relevant model class. 




\paragraph{Proofs of Proximity for Hamming Weight.} Verifying the Hamming weight of a string is one of the canonical problems considered in multiple works about proofs of proximity. In their paper introducing $\IPP$, \cite{STOC:RotVadWig13} proposed a dedicated protocol for this problem, with verifier's query and communication complexity of $\Otilde(1/\eps)$ but superlinear honest prover's query complexity. The problem was then studied in several proof models, such as non-interactive proofs of proximity (MAPs)~\cite{ITCS:GurRot15}, distribution-free proofs of proximity~\cite{aaronson_et_al:LIPIcs.CCC.2024.24} and MAPs, PCP and IOP with one-sided error~\cite{TCC:ArnBenYog24}. \cite{ITCS:AmiGolRot25} was the first and only work to our knowledge to provide a doubly-sublinear $\IPP$ for the Hamming Weight problem. Their $r$-round protocol achieves a verifier's query complexity in $O(1/\eps)$, an honest prover's query and runtime complexity in $\Otilde(n^{1/r}\cdot r^3/\eps^3)$, and a verifier's runtime and communication complexity in $O(n^{1/r}\cdot r \cdot \Otilde (1/\eps))$. 

\paragraph{Proof Systems in Machine Learning.} 

Proof systems have recently received a lot of attention in machine learning, particularly zero-knowledge proofs~\cite{Peng_Zhao_Wang_Liao_Lin_Liu_Cao_Shi_Yang_Zhang_2026}. One of their main applications is model evaluation, including to provide guarantees of accuracy~\cite{10.1145/3460120.3485379}, fairness~\cite{Zhang_Dong_Kose_Shen_Zhang_2025}, privacy~\cite{confidential_dp_2024}, or model explanations~\cite{Yadav_Laufer_Boneh_Chaudhuri_2025}. However, so far, the field of verifiable machine learning has primarily focused on proving the computations performed during evaluation,  rather than establishing statistical guarantees about the models being evaluated. A notable exception is~\cite{Paradise_Richardson_Bengio_Goldwasser_2026}, which gives an interactive PCP allowing a polynomial-time verifier to check that the conditional-probability claims of a predictive model, given as circuits, are approximately consistent with some joint distribution. This is related in spirit to our calibration criterion (Section~\ref{sec:audit-criteria}), but they quantify over all distributions and use the model's circuit description, whereas we consider sample or oracle access to a fixed distribution and target doubly-sublinear protocols. Self-Proving models~\cite{Amit_Goldwasser_Paradise_Rothblum_2025} take a complementary approach, in which the model itself is trained to act as the prover; in our setting, the honest prover only performs $O(1/g^2)$ evaluations, so the model owner can produce the proof directly. The debate line of work in AI safety~\cite{Irving_Christiano_Amodei_2018,Brown-Cohen_Irving_Piliouras_2024} similarly counts queries to an expensive oracle, namely human judgment, separately from computation, as we do with $\phi$. Returning to a single prover, \cite{Chen_Kalai_Xi_2026} build doubly-efficient interactive proofs for oracle-aided computations that are robust to a small fraction of incorrect oracle answers, an assumption reminiscent of our evaluator-disagreement model, although they verify a computation rather than a statistical property. Moreover, proofs about the computations of an evaluation are very difficult to generalize to statistical guarantees about the model: the resulting evaluations can be manipulated through the training and evaluation data, the choice of metrics, and the choice of hyperparameters~\cite{Luck_Franzese_Masserova_Takahashi_Polychroniadou_2025, dont_trust_the_process_2026}.

\paragraph{Concurrent Work.} In a concurrent and independent work, \cite{ECCC:AmiRotGol26} study interactive proofs of proximity and PAC verification in which the prover and the verifier access the input through separate noisy views, under persistent random noise. This is close in motivation to our evaluator-disagreement model, where the two parties may obtain different results when evaluating the same input. They show that such noisy $\IPP$s exist for every language in $\mathsf{NC}$ when the honest prover is clean and the noise rate is known, while in the other regimes some languages in $\mathsf{NC}^1$ admit no noisy $\IPP$ under a standard cryptographic assumption. Our focus is narrower, since we consider a single concrete problem, but our disagreement model is adversarial rather than random: the two evaluators may disagree arbitrarily on a $\rho$-fraction of the distribution and by up to $\gamma$ elsewhere, which our protocols tolerate whenever the gap exceeds twice the mean mismatch (see Section~\ref{sec:corrupted_queries}).

%% file: sections/prelims.tex
\subsection{Notation and Definitions}
\label{ssec:notation}
For an integer $n\ge 1$, we use $[n]$ to denote the set $\{1,\ldots,n\}$. Given $\mathcal{I}=(\mathcal{I}_n)_{n\in\mathbb{N}}$ a family of
input domains where inputs in $\mathcal{I}_n$ have size parameter $n$, $\Pi=(\Pi_n)_{n\in\mathbb{N}}$ is a property if $\Pi_n\subseteq\mathcal{I}_n, \forall n \in \mathbb{N}$. We write $\mathscr{P}=(\mathscr{P}_n)_{n\in\mathbb{N}}$ the set of all properties $\Pi$. Unless another distance is specified, the distance between a string $x\in\{0,1\}^n$ and a property $\Pi_n$ is the smallest relative Hamming distance between $x$ and any string $y\in\Pi_n$, i.e., $\dist(x,\Pi_n)=\min_{y\in\Pi_n} 1/n\bigl|\{i\in[n]:x_i\neq y_i\}\bigr|$, and the distance between a distribution $D$ and a property $\Pi_n$ is the smallest total variation distance between $D$ and any distribution in $\Pi_n$, i.e., $\dist(D,\Pi_n)=\min_{D'\in\Pi_n} 1/2\sum_{x\in\mathcal{X}}|D(x) - D'(x)|$.

For a bit string $x\in\{0,1\}^n$, we denote respectively its absolute and relative Hamming weights by:
\[
    \lVert x\rVert_1=\sum_{i=1}^n x_i
    \qquad\text{and}\qquad
    \wt(x)=\frac{\lVert x\rVert_1}{n}
\]
For the distribution $D$ over $\mathcal{X}$ and the function $f:\mathcal{X} \to \mathcal{Y}$, we denote $D_f$ the pushforward distribution of $D$ under $f$, i.e., $D_f(y)=\Pr_{X\sim D}[f(X)=y], \forall y\in \mathcal{Y}$. Given a distribution $D$ and a function $f:\mathcal{X}\to\{0,1\}$, we write 
\[
    \wt_D(f) := \Pr_{X\sim D}[f(X)=1] = \mathbb{E}_{X\sim D}[f(X)] =\mathbb{E}_{Y\sim D_f}[Y]
\]
for the $D$-weighted Hamming Weight of $f$.


Pairwise and, more generally, $t$-wise independent hash functions were introduced by Carter and Wegman~\cite{CarWeg79,WegCar81}. We use the definition below, and the explicit construction of~\cite{STOC:HarSah24} for arbitrary $m$ and $n$:

\begin{definition}
\label{def:t-independent}
($t$-independent function \cite{STOC:HarSah24}) For any $t\leq n$, the function $h\colon[m] \to [n]$ is \emph{$t$-independent} if $\Pr\left[ h(a_1)=\alpha_1 \land \ldots \land h(a_t)=\alpha_t \right] = 1/n^t$ for all distinct $a_1,\ldots,a_t\in[m]$, and for all $\alpha_1,\ldots,\alpha_t\in[n]$.
\end{definition}

All logarithms are to base two unless stated otherwise. The notation
$\Otilde(\cdot)$ suppresses polylogarithmic factors.

\subsection{Interactive Proofs of Proximity}
\label{ssec:ipp-model}

An \emph{interactive proof of proximity} involves a probabilistic verifier $\Verifier$ and a prover $\Prover$. The instance size and all problem parameters are given explicitly to both parties. The instance is accessed through the oracle or sampling interface specified by the problem.

In the standard string setting, an input $x\in\{0,1\}^n$ is given
through oracle access, meaning that a query $i\in[n]$ returns $x_i$.
Section~\ref{sec:distribution_weighted_hamming_weight} additionally considers instances consisting of an unknown
distribution $D$ together with a function $f$, where the parties have
sampling access to $D$ and query access to $f$.

We use the designated-honest-prover model of~\cite{ITCS:AmiGolRot25}. In this model, the protocol specifies a particular honest strategy $P_h$ that, like the verifier, receives only the access to the instance specified by the problem. Its query and sample complexity is treated as a resource of the protocol. Soundness, however, is required against every cheating prover, including one that is computationally unbounded and is given the entire input explicitly. 

\begin{definition}[Tolerant Interactive Proof of Proximity]
\label{def:tolerant-ipp}
Let $\mathcal{I}=(\mathcal{I}_n)_{n\in\mathbb{N}}$ be a family of
input domains where inputs in $\mathcal{I}_n$ have size parameter $n$, $\Pi=(\Pi_n)_{n\in\mathbb{N}}\in \mathscr{P}$ a property, $\Delta:\mathcal{I}_n \times \mathscr{P}_n \to \mathbb{R}_{\geq 0}$ be a problem-specific distance function,
and $0\leq\epsc<\epsf\leq 1$. An $(\epsc,\epsf)$-interactive proof of
proximity with respect to $\Delta$ consists of a probabilistic oracle
verifier $\Verifier$ and a designated probabilistic oracle prover $\Prover_h$
satisfying the following conditions for every $n$, every
$I\in\mathcal{I}_n$.
\begin{itemize}
    \item \emph{Completeness.} If $\Delta(I,\Pi_n)\leq\epsc$, then
    \[
    \Pr \left[
        \langle P_h^I,V^I\rangle(\Pi_n)\text{ accepts}
    \right]
    \geq \frac{2}{3},
\]
where the probability is over the randomness of both honest parties
and any randomness in the specified oracle-access model.
\item \emph{Soundness.} If $\Delta(I,\Pi_n)>\epsf$ then for every prover strategy $P^*$,
\[
    \Pr \left[
        \langle P^*,V^I\rangle(\Pi_n)\text{ accepts}
    \right]
    \leq \frac{1}{3}.
\]
The cheating prover may be computationally unbounded and nonuniform,
and may have complete knowledge of the instance.
\end{itemize}
No guarantee is required when $\epsc<\Delta(I,\Pi_n)\leq\epsf$. The quantity $g=\epsf-\epsc$ is called the tolerance gap.

The case $\epsc=0$ is called non-tolerant; in this case we write
$\eps=\epsf$. Completeness is perfect if the designated honest prover
causes the verifier to accept every pair $(I,\Pi_n)$ satisfying
$\Delta(I,\Pi_n)\leq\epsc$ with probability one.
\end{definition}


A \emph{property tester} is the special case of an IPP with no prover. An
IPP is \emph{doubly sublinear} in a given parameter regime if both the
verifier and the designated honest prover make $o(n)$ oracle queries in
that regime. The query bound $Q_P$ of the designated honest prover is required only on completeness instances. There is no query bound on a cheating prover. 

We count the verifier's and designated honest prover's oracle queries separately. Verifier bounds are worst-case over promise inputs, verifier coins, and prover strategies; designated-honest-prover bounds are worst-case over completeness instances and honest-party coins, unless stated otherwise. Communication is the total number of bits sent in both directions, including any communicated public-coin seed.

%% file: sections/hamming_weight.tex
In this section, we give a tolerant $\dsIPP$ for the Hamming Weight problem ($\HW$) whose honest prover has the optimal quadratic dependence on the tolerance gap, up to logarithmic factors in the other resources. In the non-tolerant setting,
this improves the cubic dependence on $1/\eps$ of the honest-prover complexity in~\cite{ITCS:AmiGolRot25}. We first present the protocol and then prove matching query lower bounds.

We stress that these query complexities alone are achieved by a simple one-round spot-checking baseline: the verifier sends a description of $m=O(1/g^2)$ pairwise-independent positions of $x$, the prover returns the $m$ corresponding bits, and the verifier checks that their mean is close to $\sigma/n$ and compares $O(1/g)$ randomly chosen bits against $x$. This baseline, however, requires $\Theta(1/g^2)$ bits of prover-to-verifier communication. The protocol of this section matches its query complexities with $\Otilde(1/g)$ communication. While this saving is modest in itself, it is what makes the protocol efficient in all resources in the gray-box setting of Section~\ref{subsec:whw_gray_box}, where the sample is described by a short seed.


\subsection{The Sample-And-Verify Protocol}
\label{subsec:sample-and-verify}

Our protocol is simple: the verifier first uses a pairwise-independent function to define a sample of $m=\Theta(1/g^2)$ bits of the input string. The prover computes the Hamming weight of the resulting virtual string
$y\in\{0,1\}^m$, and the parties then use an IPP for exact Hamming-weight verification to check the prover's claim about $y$. We denote this subprotocol by $\Pi_{\mathrm{HW}}$. Figure~\ref{fig:dsIPP_hamming_weight} describes the protocol, which we denote $\PiSV$.

\begin{protofig}
    {Sample-And-Verify Protocol}
    {fig:dsIPP_hamming_weight}
\emph{Query Access Input:} $x\in\{0,1\}^n$.\\
\emph{Deviation Parameters:} $\epsc\geq0$, $\epsf>\epsc$ and $g=\epsf-\epsc$.\\
\emph{Other Parameters:} Claimed (absolute) weight $\sigma \in \N$, number of samples $m=\Theta(1/g^2)$.

\begin{enumerate}[itemsep=2pt,topsep=3pt]
\item $\Verifier \to \Prover$: send a description of a $2$-independent function $h\colon[m] \to [n]$. 

Define $y\in\{0,1\}^m$ such that $y_i=x_{h(i)}, \forall i\in[m]$.
\item $\Prover \to \Verifier$: send a claimed absolute weight
$\sigma_y\in\{0,\ldots,m\}$. The honest prover sets $\sigma_y=\lVert y\rVert_1$.
\item $\Verifier$: receive $\sigma_y$, reject if $\bigl|\sigma_y/m-\sigma/n\bigr|>g/4+\epsc$ or if $\sigma_y\notin\{0,\ldots,m\}$.
\item $\Prover$ and $\Verifier$: run $\Pi_{\mathrm{HW}}$ on input $y$, with claimed absolute weight $\sigma_y$ and proximity parameter $g/4$. The verifier outputs the subprotocol's verdict.
\end{enumerate}
\end{protofig}

\begin{theorem}[Doubly-sublinear Interactive Proof of Proximity for Hamming Weight]
\label{thm:dsIPP_hw}
Let $0\leq\epsc<\epsf\leq1$ and let $g=\epsf-\epsc$.
Let $\Pi_{\mathrm{HW}}$ be an $\IPP$ for exact Hamming-weight verification with proximity parameter $g/4$. On inputs of length $N$, denote its verifier query complexity by $\mathrm{Q_\Verifier^{\Pi_{\mathrm{HW}}}}(N)$, verifier time complexity by $\mathrm{T_\Verifier^{\Pi_{\mathrm{HW}}}}(N)$, communication complexity by $\mathrm{C^{\Pi_{\mathrm{HW}}}}(N)$, prover time complexity by $\mathrm{T_\Prover^{\Pi_{\mathrm{HW}}}}(N)$, and number of rounds by $r^{\Pi_{\mathrm{HW}}}(N)$. For $m=\Theta(1/g^2)$, the Sample-And-Verify protocol is a $(1+r^{\Pi_{\mathrm{HW}}}(m))$-round $(\epsc,\epsf)$-$\IPP$ for Hamming-weight verification with:
\begin{itemize}
    \item Verifier query complexity $\mathrm{Q_\Verifier^{\Pi_{\mathrm{HW}}}}(m)$. 
    \item Honest-prover query complexity at most $m$;
    \item Communication complexity
    $O(\log n)+\mathrm{C^{\Pi_{\mathrm{HW}}}}(m)$;
\end{itemize}
For every fixed $\lambda>0$, the verifier's expected time complexity is
\(O(n^\lambda)+\mathrm{T_\Verifier^{\Pi_{\mathrm{HW}}}}(m)\),
and, assuming the Generalized Riemann Hypothesis, it is
\(O(\log n)+\mathrm{T_\Verifier^{\Pi_{\mathrm{HW}}}}(m)\).
The honest prover's time complexity is
\(O(m+\log n)+\mathrm{T_\Prover^{\Pi_{\mathrm{HW}}}}(m)\).
In the regime $1/g^2=o(n)$, the protocol is doubly sublinear.
\end{theorem}

\begin{proof}
If necessary, we apply constant amplification to $\Pi_{\mathrm{HW}}$ so that its completeness and soundness errors are sufficiently small constants.

\textbf{Completeness.} The $m$ samples defined by the $2$-independent function are pairwise independent and uniformly distributed (see Definition~\ref{def:t-independent}). Completeness follows from the observation that sampling $O(1/g^2)$ pairwise-independent and uniformly distributed bits of $x$ is sufficient to ensure its Hamming weight is close to the claimed weight w.h.p. To demonstrate this, we first note that, by triangle inequality:
\begin{align*}
\left| \wt(y)-\sigma/n\right| &\leq \left| \wt(y)-\wt(x)\right| + \left| \sigma/n-\wt(x)\right|\\
&\leq \left| \wt(y)-\wt(x)\right| + \epsc
\end{align*}
Therefore
\[
\Pr\left[ \left| \wt(y)-\sigma/n\right|>g/4+\epsc\right]\leq\Pr\left[ \left| \wt(y)-\wt(x)\right|>g/4\right]
\]
Since the sampled bits are pairwise independent and uniform,
\[
    \mathbb{E}[\wt(y)]=\wt(x)
    \quad\text{and}\quad
    \mathrm{Var}[\wt(y)]
    =
    \frac{1}{m^2}\sum_{j=1}^m \mathrm{Var}[y_j]
    \leq \frac{1}{4m}.
\]
By Chebyshev's inequality,
\[
    \Pr\left[
        |\wt(y)-\wt(x)|>\frac{g}{4}
    \right]
    \leq \frac{4}{mg^2}.
\]
Therefore, $m=4/(\delta g^2)$ samples are sufficient to ensure that $\Pr\left[ \left| \wt(y)-\sigma/n\right|\leq g/4+\epsc\right]>1-\delta$ (i.e., that $\Verifier$ does not reject in step 3) if the claimed weight $\sigma$ is $\epsc$-close to the actual weight. In step 4, the honest prover always provides a YES instance to the subprotocol $\Pi_{\mathrm{HW}}$. Hence the overall completeness error is at most $\delta$ plus the completeness error of $\Pi_{\mathrm{HW}}$, which is at most $1/3$ by taking both to be sufficiently small constants.

\textbf{Soundness.} Let $m=4/(\delta g^2)$ for some $\delta>0$. For any input $x$ such that $|\wt(x)-\sigma/n| > \epsf$ and any prover $\Prover^*$, the verifier accepts only if $|\sigma_y/m-\sigma/n|\leq g/4 + \epsc$ (step 3) and if the verifier in the subprotocol accepts (step 4). In the case where $|\wt(y)-\sigma_y/m| > g/4$, the probability that the verifier of the subprotocol accepts is at most $\delta_{\Pi_{\mathrm{HW}}}$ (i.e., the soundness error of the subprotocol). 
Therefore, for any NO instance $x$ and any $\Prover^*$:
\[
\Pr \left[ \Verifier \ \mathrm{accepts} \ x \right] \leq \Pr \left[ |\wt(y)-\sigma_y/m| \leq g/4\right] + \delta_{\Pi_{\mathrm{HW}}}
\]
We first note that, by the triangle inequality:
\[
|\wt(y)-\sigma_y/m| \geq |\wt(x)-\sigma/n| - |\wt(x)-\wt(y)| - |\sigma_y/m-\sigma/n| 
\]
Moreover, independently of $\Prover^*$, by Chebyshev's inequality:
\[
\Pr\left[ \left| \wt(x)-\wt(y)\right|>g/2\right] \leq \delta/4
\]
By combining these two results:
\begin{align*}
    \Pr [ |\wt(y)-\sigma_y/m| \leq g/4] &\leq \Pr [ |\wt(x)-\sigma/n| - |\wt(x)-\wt(y)| - |\sigma_y/m-\sigma/n| \leq g/4]\\
    &\leq \Pr [ \epsf - |\wt(x)-\wt(y)| - g/4 - \epsc \leq g/4]\\
    &\leq \Pr [|\wt(x)-\wt(y)| \geq g/2]\\
    &\leq \delta/4
\end{align*}
Therefore, the soundness error is at most \(\frac{\delta}{4}+\delta_{\Pi_{\mathrm{HW}}}\), which is at most $1/3$ by taking $\delta$ and the soundness error of $\Pi_{\mathrm{HW}}$ to be sufficiently small constants.

\textbf{Complexity.}
The communication complexity follows from the fact that, for all $m$ and $n$, there exists a $2$-independent function $h\colon[m]\to[n]$ that can be represented using $O(\log n+\log m)$ bits, as shown in~\cite{STOC:HarSah24}. In the regime considered here, $m<n$, so this is $O(\log n)$ bits.

Using the \textsc{Constructor} of~\cite{STOC:HarSah24} with parameters $(n,m,2)$, the verifier can construct such a function in expected time $O(n^\lambda)$ for every fixed $\lambda>0$. Assuming the Generalized Riemann Hypothesis, the construction time is $O(\log n)$.

Regarding the query complexity, the prover does not need to query further elements of the input in the subprotocol $\Pi_{\mathrm{HW}}$ given that it has already read the entire string $y$ in step 2 to compute the exact Hamming weight of $y$. The verifier does not need to query any bit of the input prior to executing the subprotocol $\Pi_{\mathrm{HW}}$ in step 4.
\end{proof}

We obtain the optimal quadratic dependence on the tolerance gap by instantiating $\Pi_{\mathrm{HW}}$ in Theorem~\ref{thm:dsIPP_hw} with the Hamming-weight $\dsIPP$ of~\cite{ITCS:AmiGolRot25}. Table~\ref{tab:hw-comparison} compares our instantiations with \cite{ITCS:AmiGolRot25}.

\begin{corollary}[Logarithmic-round $\dsIPP$ for Hamming weight]
\label{cor:log-round-hw}
Let $0\leq\epsc<\epsf\leq1$ and let $g=\epsf-\epsc$. There exists a $\dsIPP$ for the Hamming Weight problem with $O(\log(1/g))$ rounds, verifier query complexity $\ O(1/g)$, honest-prover query complexity $\ O(1/g^2)$, and communication complexity $\ O(\log^2(1/g)/g + \log n)$.
\end{corollary}

\begin{corollary}[Constant-round $\dsIPP$ for Hamming weight]
\label{cor:cst-round-hw}
Let $0\leq\epsc<\epsf\leq1$ and let $g=\epsf-\epsc$. There exists a $\dsIPP$ for the Hamming Weight problem with $r$ rounds, verifier query complexity $\ O(1/g)$, honest-prover query complexity $\ O(1/g^2)$, and communication complexity $\ O((1/g^2)^{1/r} \cdot \log(1/g)/g + \log n)$.
\end{corollary}


\begin{remark}
\cite{ITCS:AmiGolRot25} also takes the approach of approximating the Hamming weight to achieve a doubly sublinear protocol. However, rather than first approximating the Hamming weight of the entire input string $X$, the prover in their protocol approximates the Hamming weight of each row in the matrix representation of $X$. In our solution, the prover does not need to query the input again after obtaining the initial approximation of its Hamming weight, since the subsequent subprotocol is run only on the sampled queries. In contrast, in~\cite{ITCS:AmiGolRot25}, the prover must query the input to compute each approximation, across each of the $O(1/\eps)$ repetitions, explaining the cubic dependence in their prover's query complexity.
\end{remark}

\subsubsection{One-Round Protocol}

The Sample-And-Verify transformation can be instantiated with the one-round Hamming-weight protocol of~\cite{STOC:RotVadWig13}, however the naive instantiation leads to a $2$-round protocol. To reduce the protocol to one round, the verifier sends the description of the sampled virtual string together with the first message of the underlying Hamming-weight protocol. This preserves a single round of interaction, at the cost of a larger verifier query complexity.

Concretely, the one-round protocol of~\cite{STOC:RotVadWig13} views its input as a $k\times\ell$ matrix, with $k=m^{2/3}$ and $\ell=m^{1/3}$. After a verifier-chosen pseudorandom permutation of the coordinates, the prover reports the Hamming weight of every row, and the verifier checks a suitable random subset of these claims. We apply this protocol to the sampled string $y$ produced by our Sample-And-Verify transformation. Figure~\ref{fig:one_round_dsIPP_hamming_weight} describes the resulting one-round protocol.

\begin{protofig}
    {One-Round Sample-And-Verify Protocol}
    {fig:one_round_dsIPP_hamming_weight}

\emph{Query Access Input:} $x\in\{0,1\}^n$.\\
\emph{Deviation Parameters:} $0\leq\epsc<\epsf\leq1$ and
$g=\epsf-\epsc$.\\
\emph{Other Parameters:} Claimed absolute weight
$\sigma\in\{0,\ldots,n\}$ and number of samples
$m=\Theta(1/g^2)$.

\begin{enumerate}[itemsep=2pt,topsep=3pt]

\item $\Verifier\to\Prover$:
sample a $2$-independent function $h\colon[m]\to[n]$ and the verifier's
first-message randomness for the one-round Hamming-weight protocol
of~\cite{STOC:RotVadWig13}, instantiated on inputs of length $m$ with
proximity parameter $g/4$. Send both descriptions to the prover.

Define the virtual string $y\in\{0,1\}^m$ by \(y_i=x_{h(i)}\), \( i\in[m]\).

\item $\Prover\to\Verifier$:
send a claimed absolute weight $\sigma_y\in\{0,\ldots,m\}$ together with
the prover message of the one-round Hamming-weight protocol
of~\cite{STOC:RotVadWig13} for the claim that $y$ has Hamming weight
$\sigma_y$. The honest prover sets
\[
    \sigma_y=\|y\|_1.
\]

\item $\Verifier$:
reject if
\[
    \left|\frac{\sigma_y}{m}-\frac{\sigma}{n}\right|
    >
    \epsc+\frac{g}{4}.
\]
Otherwise, execute the verifier checks of the one-round Hamming-weight
protocol on $y$. Whenever that verifier queries position $i$ of $y$,
query $x_{h(i)}$. Output the resulting verdict.

\end{enumerate}
\end{protofig}

The following theorem gives the resulting complexities. As in Theorem~\ref{thm:dsIPP_hw}, the honest prover reads the entire sampled string $y$, while the verifier accesses only the coordinates queried by the underlying Hamming-weight protocol.

\begin{theorem}[One-round $\dsIPP$ for Hamming weight]
\label{thm:one-round-hw}
There exists a one-round $(\epsc,\epsf)$-IPP for Hamming-weight verification with verifier query complexity \(\Otilde(1/g^{4/3})\), honest-prover query complexity \(O(1/g^2)\), and communication complexity \(O(\log n)+\Otilde(1/g^{4/3})\). In the regime $1/g^2=o(n)$, the protocol is doubly sublinear.
\end{theorem}

\begin{proof}
The completeness and soundness arguments are the same as in Theorem~\ref{thm:dsIPP_hw}; the only difference is the choice of the exact-Hamming-weight subprotocol.

Let $d$ denote the absolute proximity parameter of the Hamming-weight subprotocol. Since its input $y$ has length $m$ and its relative proximity parameter is $g/4$, \(d=gm/4\). The one-round Hamming-weight protocol of~\cite{STOC:RotVadWig13} has query and communication complexities \(\Otilde\left(m/(d^{2/3})\right)\) on an $m$-bit input with absolute proximity parameter $d$. Substituting $d=gm/4$ gives
\[
    \Otilde\left(\frac{m}{(gm)^{2/3}}\right)
    =
    \Otilde\left(\frac{m^{1/3}}{g^{2/3}}\right).
\]
Since $m=\Theta(1/g^2)$,
\[
    \frac{m^{1/3}}{g^{2/3}}
    =
    \Theta(1/g^{4/3}),
\]
which proves the claimed verifier query and subprotocol communication bounds.

The verifier's first message in the one-round protocol of~\cite{STOC:RotVadWig13} is independent of the prover's message and can therefore be sent together with the description of $h$. The prover can then send $\sigma_y$ together with its response in the underlying one-round Hamming-weight protocol. Hence the composed protocol still has one verifier message followed by one prover message.

The additional communication required to describe $h$ is $O(\log n)$, as in Theorem~\ref{thm:dsIPP_hw}. Thus the total communication complexity
is
\[
    O(\log n)+\Otilde(1/g^{4/3}).
\]

Finally, the honest prover constructs $y$ by querying at most its $m$ underlying coordinates and stores these values. It can therefore perform the computations required by the Hamming-weight subprotocol without making any additional queries to the original input. Its query complexity is consequently at most
\[
    m=O(1/g^2).
\]

\end{proof}

\begin{table}[tbp]
    \centering
    \begin{tabular}{l @{\hspace{1.2\tabcolsep}} c @{\hspace{1.2\tabcolsep}} c @{\hspace{1.2\tabcolsep}} c @{\hspace{1.2\tabcolsep}} c}
        \toprule
        &Verifier's Queries&Prover's Queries&Communication&Rounds\\
        \midrule
        \cite{ITCS:AmiGolRot25}&$O\left(1/\eps\right)$&$\Otilde \left({n}^{1/r(n)}\cdot r(n)^3/\eps^3 \right) $&$O\left({n}^{1/r(n)} \cdot r(n) \cdot \log(r(n)/\eps)/\eps\right)$&$r(n)$\\
        Corollary~\ref{cor:log-round-hw}&$O(1/\eps)$&$O(1/\eps^2)$&$\Otilde(1/\eps) + O( \log n)$&$O(\log (1/\eps))$\\
        Corollary~\ref{cor:cst-round-hw}&$O(1/\eps)$&$O(1/\eps^2)$&$\Otilde(1/\eps^{1+2/r}) + O( \log n)$&$r$\\
        Theorem~\ref{thm:one-round-hw}&$\Otilde(1/\eps^{4/3})$&$O(1/\eps^2)$&$\Otilde(1/\eps^{4/3})+O(\log n)$&$1$\\
        \bottomrule
    \end{tabular}
    \caption{Comparison of the non-tolerant version of our logarithmic-round and one-round $\dsIPP$s with the $\dsIPP$ of \cite{ITCS:AmiGolRot25}. $\eps$ denotes the proximity parameter.}
    \label{tab:hw-comparison}
\end{table}


\subsection{Optimality}

In this subsection, we demonstrate the optimality of the logarithmic-round instantiation of our protocol (see Corollary~\ref{cor:log-round-hw}), by proving the matching query lower bound. We then prove that there exists no $\dsIPP$ for the Hamming weight problem that achieves perfect completeness. For $\sigma$ with $\sigma n\in\{0,\ldots,n\}$, write $\HW_n(\sigma):=\{x\in\{0,1\}^n:\wt(x)=\sigma\}$.


\begin{proposition}[Query complexity lower bounds]
\label{prop:hw-query-lower-bounds}
Let \(0\leq\epsc<\epsf\leq\frac14\), \(g:=\epsf-\epsc\) and suppose that $1/g^2\leq n$.
Every $(\epsc,\epsf)$-IPP for Hamming-weight verification satisfies
\(Q_\Verifier=\Omega(1/g)\) and \(Q_\Verifier+Q_\Prover=\Omega(1/g^2)\).
In particular, if $Q_\Verifier=o(1/g^2)$, then $Q_\Prover=\Omega(1/g^2)$.
\end{proposition}

\begin{proof}
We ignore rounding throughout; rounding the relevant weights changes them by $O(1/n)$, which does not affect the asymptotic bounds in the regime $1/g^2\leq n$.

We first prove the verifier lower bound. Fix the claimed absolute weight $\sigma=n/2$ and a string $x$ such that
\[
    \wt(x)=\frac12+\epsc.
\]
Thus $x$ is a completeness instance.

Choose uniformly at random a set $I$ of $3gn/2$ coordinates among the zero positions of $x$, and define
\(x^I:=x\oplus\mathbf{1}_I\).
Then
\[
    \wt(x^I)-\frac12
    =
    \epsc+\frac{3g}{2}
    =
    \epsf+\frac{g}{2}
    >
    \epsf,
\]
so $x^I$ is a soundness instance.

On input $x^I$, let $\Prover_I^*$ be the cheating prover that emulates the designated honest prover on the fixed input $x$. Couple this execution with the honest execution on $x$ using the same verifier and prover randomness. Condition on these coins, and let $Q\subseteq[n]$ be the coordinates queried by the verifier in the honest execution.

Since $\epsc\leq1/4$, the string $x$ has at least $n/4$ zero positions.
Therefore
\[
    \Pr_I[I\cap Q\neq\emptyset]
    \leq
    |Q|\frac{3gn/2}{n/4}
    =
    6g|Q|.
\]
Whenever $I\cap Q=\emptyset$, the verifier receives exactly the same
oracle answers and prover messages in the two executions. Hence
\[
    \mathbb{E}_I \left[
        \Pr \left[
            \Verifier^{x^I}\leftrightarrow\Prover_I^*
            \text{ accepts}
        \right]
    \right]
    \geq
    \frac23-6gQ_\Verifier.
\]
By soundness, the left-hand side is at most $1/3$. Thus
\[
    Q_\Verifier=\Omega(1/g).
\]

For the combined-query lower bound, let $\widehat{\Prover}$ be the designated honest prover truncated after $Q_\Prover$ oracle queries: if its simulation attempts to make another query, it stops the simulation
and behaves arbitrarily thereafter. On every completeness instance, $\widehat{\Prover}$ behaves identically to the designated honest prover, while on every input it makes at most $Q_\Prover$ queries.

An ordinary randomized tester can simulate the interaction between $\widehat{\Prover}$ and $\Verifier$, answering every oracle query of either party using its own oracle access. The tester therefore makes at most
\[
    q:=Q_\Verifier+Q_\Prover
\]
queries. Completeness and soundness of the IPP imply that it distinguishes strings of relative weights
\[
    p_Y:=\frac12+\epsc
    \qquad\text{and}\qquad
    p_N:=\frac12+\epsc+\frac{3g}{2}
\]
with constant advantage.

We show that this requires $q=\Omega(1/g^2)$. If $q\geq n/16$, the claim already follows from $1/g^2\leq n$. Otherwise, consider a uniformly random string of each of the two prescribed weights and assume that the tester
never repeats a query. After $t-1$ queries have revealed $s$ ones, the conditional probabilities that the next answer is one are
\[
    u_t=\frac{p_Yn-s}{n-t+1},
    \qquad
    v_t=\frac{p_Nn-s}{n-t+1}.
\]
Since $q<n/16$ and
\[
    \frac12\leq p_Y<p_N\leq\frac78,
\]
both probabilities remain bounded away from $0$ and $1$ by universal
constants, while
\[
    |u_t-v_t|=O(g).
\]
Hence
\[
    D_{\mathrm{KL}}\bigl(\Ber(u_t)\,\|\,\Ber(v_t)\bigr)=O(g^2).
\]
By the chain rule for relative entropy, the KL divergence between the two complete tester views is $O(qg^2)$. Pinsker's inequality therefore bounds their total variation distance by $O(g\sqrt q)$.

The acceptance probabilities in the two cases differ by at least $1/3$, so their total variation distance is at least $1/3$. Consequently
\[
    q=\Omega(1/g^2).
\]
\end{proof}


\begin{proposition}[No perfect completeness for dsIPPs]
\label{prop:no-perfect}
Let $\sigma=\sigma(n)\in(0,1)$ satisfy \(\min\{\sigma,1-\sigma\}=\Omega(1)\) and \(\sigma n\in\{0,\dots,n\}\).
Fix constant parameters
\(0\leq\epsc<\epsf\leq1\)
such that
\( \epsf<\max\{\sigma,1-\sigma\}\).
Then every $(\epsc,\epsf)$-IPP for $\HW_n(\sigma)$ with perfect completeness and constant soundness error satisfies
\[
    Q_\Verifier+Q_\Prover=\Omega(n).
\]
In particular, no such IPP is doubly sublinear.
\end{proposition}

\begin{proof} Suppose toward a contradiction that there is an IPP for $\HW_n(\sigma)$ with perfect completeness, constant soundness error, and 
\[ q(n):=Q_\Verifier(n)+Q_\Prover(n)=o(n). \] 
Let $\Prover_{\mathrm{hon}}$ be the prescribed honest prover. To avoid relying on its behaviour or query complexity outside the yes-instances, define a truncated prover $\widehat{\Prover}$ as follows: it simulates $\Prover_{\mathrm{hon}}$, but if the simulation attempts to make more than $Q_\Prover(n)$ oracle queries, it stops the simulation and behaves arbitrarily thereafter. On every yes-instance, $\widehat{\Prover}$ behaves identically to $\Prover_{\mathrm{hon}}$, while on every input it makes at most $Q_\Prover(n)$ queries. 

Define a prover-free randomized tester $T$ that internally simulates the interaction between $\widehat{\Prover}$ and $\Verifier$. Whenever either simulated party queries a coordinate of the oracle input, $T$ makes the same query and returns the answer to that party. Thus, $T$ makes at most \[ Q_\Verifier(n)+Q_\Prover(n)=q(n) \] 
queries. If $x\in\HW_n(\sigma)$, then $\widehat{\Prover}$ behaves exactly as the honest prover, so perfect completeness implies that $T$ accepts with probability $1$. If $\dist(x,\HW_n(\sigma))>\epsf$, then $\widehat{\Prover}$ is one particular legal prover strategy. Soundness therefore implies that $T$ rejects with constant probability. Hence $T$ is a one-sided-error tester for $\HW_n(\sigma)$ using $q(n)=o(n)$ queries. 

Since there are only finitely many strings of relative Hamming weight $\sigma$, perfect completeness also implies that there is a probability-one set $\mathcal{G}$ of random tapes such that, for every $\rho\in\mathcal{G}$, the tester $T$ accepts every string $x\in\HW_n(\sigma)$ when run with random tape $\rho$.

We now derive a contradiction directly. Since
\[
    \epsf<\max\{\sigma,1-\sigma\},
\]
at least one of the constant strings $0^n$ and $1^n$ is a soundness instance. We consider the two cases separately. 

\noindent \emph{Case 1: $\sigma>\epsf$.}
Then 
\[
\dist\bigl(0^n,\HW_n(\sigma)\bigr)=\sigma>\epsf.
\]
By soundness, $T$ rejects $0^n$ with positive constant probability. Since $\mathcal{G}$ has probability one, we may fix
$\rho\in\mathcal{G}$ for which $T$ rejects $0^n$, and let $Q\subseteq[n]$ be the set of coordinates queried in this execution. Since $|Q|\le q(n)=o(n)$ and $1-\sigma=\Omega(1)$, for all sufficiently large $n$, \[ |Q|<(1-\sigma)n. \] 
Consequently, \[ |[n]\setminus Q|>\sigma n. \] 
Choose a set \[ R\subseteq[n]\setminus Q \qquad\text{such that}\qquad |R|=\sigma n, \] and define $x'\in\{0,1\}^n$ by 
\[ x'_i= \begin{cases} 1 & \text{if } i\in R,\\ 0 & \text{otherwise}. \end{cases} \] 
Then \[ \wt(x')=\frac{|R|}{n}=\sigma, \] so $x'\in\HW_n(\sigma)$. Moreover, on the fixed random coins $\rho$, the tester has exactly the same execution on $x'$ as on $0^n$. Indeed, every queried coordinate lies in $Q$, and $x'_i=0$ for every $i\in Q$. By induction over the adaptive queries, the tester receives the same answers, makes the same subsequent queries, and eventually rejects. Thus $T$ rejects $x'$ on the random tape $\rho$. But
$\rho\in\mathcal{G}$ and $x'\in\HW_n(\sigma)$, contradicting the definition of $\mathcal{G}$.

\noindent \emph{Case 2: $1-\sigma>\epsf$.}
Then 
\[
\dist\bigl(1^n,\HW_n(\sigma)\bigr)=1-\sigma>\epsf.
\]
By the same argument, fix $\rho\in\mathcal{G}$ for which $T$ rejects $1^n$, and let $Q\subseteq[n]$ be the queried coordinates. Since $|Q|\le q(n)=o(n)$ and $\sigma=\Omega(1)$, for all sufficiently large $n$, 
\[ |Q|<\sigma n. \] 
Choose a set \[ R\subseteq[n]\setminus Q \qquad\text{such that}\qquad |R|=\sigma n-|Q|, \] 
and define $x'\in\{0,1\}^n$ by 
\[ x'_i= \begin{cases} 1 & \text{if } i\in Q\cup R,\\ 0 & \text{otherwise}. \end{cases} \] 
Then \[ \|x'\|=|Q|+|R|=\sigma n, \] and hence $x'\in\HW_n(\sigma)$. Every coordinate queried in the execution on $1^n$ belongs to $Q$ and still has value $1$ in $x'$. Therefore, on the fixed coins $\rho$, the tester has the same adaptive execution on $x'$ as on $1^n$ and rejects. This contradicts $\rho\in\mathcal{G}$. In either case we obtain a contradiction. Therefore, 
\[ Q_\Verifier+Q_\Prover=\Omega(n). \] \end{proof} 

\begin{remark}[The completeness error is intrinsic]
\label{rem:perfect-completeness}
Proposition~\ref{prop:no-perfect} shows that, for an interior claimed Hamming weight and any non-vacuous constant soundness radius, perfect completeness is incompatible with
\[
    Q_\Verifier+Q_\Prover=o(n).
\]
Thus, the completeness error of \(\PiSV\) is not an artifact of its concentration argument: any perfectly complete protocol in this regime must have linear total honest-party query complexity.

This does not contradict the perfectly complete Hamming-weight proximity proofs of~\cite{TCC:ArnBenYog24}. Their verifiers are sublinear, but their honest provers read the entire input, so their total honest-party query complexity is linear.
\end{remark}



%% file: sections/unified-4-5.tex
We now study the \emph{distribution-weighted Hamming weight problem}. Let $D$ be a distribution over a finite domain $\mathcal{X}$ and let $f\colon\mathcal{X}\to\{0,1\}$ be an evaluation function. We define the \emph{$D$-weighted Hamming weight of $f$} by
\[
    \wt_D(f)
    :=
    \Pr_{X\sim D}[f(X)=1]
    =
    \mathbb{E}_{X\sim D}[f(X)].
\]
When $\mathcal{X}=[n]$ and $D$ is uniform, $\wt_D(f)$ is precisely the relative Hamming weight of the truth table of $f$.

\begin{remark}[Distance for weighted Hamming weight]
\label{rem:whw-distance}
For distributional weighted Hamming-weight verification, we instantiate Definition~\ref{def:tolerant-ipp} with $\Delta$ the total variation distance. In the case of the distributional weighted Hamming-weight, this is equivalent to the distance between the relative weight and the claimed one:
\[
    \Delta_{\dwHW}\bigl((D,f),\sigma\bigr)
    :=
    \left|\wt_D(f)-\sigma\right|.
\]

\end{remark}

We first consider the setting in which the prover and verifier query the same evaluation function, under black-box and gray-box access to the distribution. We then consider the setting in which they use different evaluation functions that satisfy a bounded disagreement promise. For this second setting, we extend the formulation to finite-precision scores
in $[0,1]$, with Boolean evaluations as a special case.

\paragraph{Access Models.}
Since the distribution $D$ is unknown, query access to the evaluation function alone does not determine the distribution-weighted Hamming weight. We therefore additionally give the parties access to $D$. We consider the following two sampling models, which will be used in both settings below.

\begin{definition}[Black-box sampling] A black-box sampling call returns a fresh independent element $X\sim D$. The caller cannot choose or inspect the randomness used to generate $X$.
\end{definition}

\begin{definition}[Gray-box sampling]
\label{def:gray-box-sampling}
A gray-box sampling call consists of the use of a deterministic sampling oracle
\[
    G\colon\{0,1\}^{R}\longrightarrow\mathcal{X}
    \qquad\text{such that}\qquad
    G(U_R)\sim D,
\]
where $U_R$ is uniform over $\{0,1\}^{R}$. The parties do not need a
description of $G$, but may invoke it on a random tape $r$ of their choice. Each invocation of $G$ counts as one sample-generation call. Evaluating the relevant evaluation function at $G(r)$ therefore costs one call to $G$ and one query to the corresponding evaluation oracle.
\end{definition}

%% file: sections/weighted_hamming_weight.tex
\subsection{Common-Evaluator Setting}
\label{sec:weighted_hamming_weight}

We begin with the setting in which the prover and verifier have query access to the same evaluation function $f$. We first give protocols for black-box and gray-box sampling. Then, we show how a public sample can be reused across several weighted Hamming-weight claims.




\subsubsection{DsIPP with Black-Box Sampling and Query Access}
\label{subsec:whw_black_box}

\paragraph{A spot-checking baseline.} In the black-box setting, the verifier must send the entire realized sample to the prover, so the following one-round protocol is a natural baseline: the prover returns all evaluations $Y_j=f(X_j)$ for $j\in[m]$, the verifier rejects if their mean is too far from $\sigma$, and otherwise checks $O(1/g)$ uniformly random positions against its own queries to $f$. A Hoeffding bound together with a standard spot-checking argument shows that this protocol uses $O(1/g)$ verifier queries to $f$, at most $m$ honest-prover queries, and $O(m\ell)$ bits of communication, where $\ell$ is the number of bits used to encode an element of $\mathcal{X}$. Protocol $\PiWM^{\mathrm{BB}}$ below matches these query complexities and improves on the baseline only by replacing the $m$-bit prover message with the communication of the Hamming-weight subprotocol of Section~\ref{sec:hamming_weight}. Since the verifier-to-prover message already has size $O(m\ell)$, this saving is minor in the black-box setting. We present $\PiWM^{\mathrm{BB}}$ mainly because the same structure becomes efficient in all resources in the gray-box setting (Section~\ref{subsec:whw_gray_box}), where the sample is described by a short seed.

\paragraph{The protocol.} Fix an error parameter $\delta\in(0,1/4)$ and set
\[
    m
    =
    \left\lceil
        \frac{8\ln(2/\delta)}{g^2}
    \right\rceil.
\]
We use the exact-Hamming-weight subprotocol from
Section~\ref{sec:hamming_weight}, amplified so that its completeness and
soundness errors are at most $\delta$. In the invocation below, its input has length $m$, its claimed absolute weight is an integer
$W\in\{0,\ldots,m\}$, and its proximity parameter is $g/4$.

\begin{protofig}
    {Protocol $\PiWM^{\mathrm{BB}}$}
    {fig:protocol-piwm-bb}

\begin{enumerate}
    \item The verifier $\Verifier$ draws
          $X_1,\ldots,X_m$ independently from $D$ and sends the ordered sample
          $S=(X_1,\ldots,X_m)$
          to the prover $\Prover$. Define the virtual string
          $Y\in\{0,1\}^{m}$ by $Y_j=f(X_j)$.

    \item The honest prover queries $f$ once on each distinct sampled point,
          caches the answers, computes $W=\lVert Y\rVert_1$, and sends $W$
          to the verifier.

    \item The verifier rejects if
          $W\notin\{0,\ldots,m\}$ or
          \[
              \left|
                  \frac{W}{m}-\sigma
              \right|
              >
              \epsc+\frac{g}{2}.
          \]

    \item The parties run the exact-Hamming-weight subprotocol from
          Section~\ref{sec:hamming_weight} on $Y$, with claimed absolute
          weight $W$ and proximity parameter $g/4$.

          Whenever the verifier of that subprotocol queries position $j$ of
          $Y$, the present verifier queries $f(X_j)$. The verifier outputs the
          subprotocol's decision.
\end{enumerate}

\end{protofig}

\begin{theorem}[Black-box upper bound]
\label{thm:black-box-whw}
Protocol $\PiWM^{\mathrm{BB}}$ in
Figure~\ref{fig:protocol-piwm-bb} has completeness and soundness errors
at most $2\delta$ for distributional weighted Hamming-weight verification.
In particular, when $\delta\leq 1/6$, it is an
$(\epsc,\epsf)$-interactive proof with respect to
$\Delta_{\dwHW}$ in the sense of
Definition~\ref{def:tolerant-ipp}.

Using the logarithmic-round exact-Hamming-weight protocol from
Section~\ref{sec:hamming_weight}, its resource bounds are as follows:
\begin{itemize}
    \item the verifier draws exactly
          $m=O(\ln(1/\delta)/g^2)$ samples from $D$ and, for every fixed
          $\delta$, makes
          $O(1/g)$ queries to $f$;

    \item the honest prover draws no samples from $D$ and makes at most
          $\min\{m,|\mathcal{X}|\}$ distinct queries to $f$;

    \item the communication is
          \[
              O(m\ell)+\Otilde(1/g)
              =
              O\left(
                  \frac{\ell\ln(1/\delta)}{g^2}
              \right)
              +\Otilde(1/g)
          \]
          bits for every fixed $\delta$, and the number of rounds is
          $O(\log(1/g))$ up to a constant additive term.
\end{itemize}
Here $\ell$ denotes the number of bits used to encode an element of
$\mathcal{X}$.
\end{theorem}

\begin{proof}
Write $\mu=\wt_D(f)$.
Since the $X_j$ are independent, the bits
$Y_j=f(X_j)$ are independent Bernoulli random variables with mean $\mu$.
Hoeffding's inequality gives
\[
\begin{split}
    \Pr\left[
        \left|
            \frac{\lVert Y\rVert_1}{m}-\mu
        \right|
        >
        \frac{g}{4}
    \right]
    &\leq
    2\exp\left(
        -2m\left(\frac{g}{4}\right)^2
    \right) \\
    &=
    2\exp\left(-\frac{mg^2}{8}\right) \\
    &\leq
    \delta.
\end{split}
\]
Let $\mathcal{E}$ denote the complementary event. The event
$\mathcal{E}$ depends only on $f$ and the verifier's sample, and not on any
message sent by the prover.

Suppose first that $|\mu-\sigma|\leq\epsc$ and the prover is honest. On
$\mathcal{E}$, its message
$W=\lVert Y\rVert_1$ satisfies
\[
\begin{split}
    \left|
        \frac{W}{m}-\sigma
    \right|
    &\leq
    \left|
        \frac{\lVert Y\rVert_1}{m}-\mu
    \right|
    +
    |\mu-\sigma| \\
    &\leq
    \frac{g}{4}+\epsc.
\end{split}
\]
The consistency check therefore passes. The weight supplied to the
exact-Hamming-weight subprotocol is correct, so the subprotocol rejects with
probability at most $\delta$. Including the probability that $\mathcal{E}$
fails, the total rejection probability is at most $2\delta$.

Now suppose that $|\mu-\sigma|>\epsf$ and fix an arbitrary cheating
prover. Condition on any sample for which
$\mathcal{E}$ holds and any message $W$ that passes the consistency check.
Then
\[
\begin{split}
    \left|
        \frac{\lVert Y\rVert_1}{m}-\frac{W}{m}
    \right|
    &\geq
    |\mu-\sigma|
    -
    \left|
        \mu-\frac{\lVert Y\rVert_1}{m}
    \right|
    -
    \left|
        \frac{W}{m}-\sigma
    \right| \\
    &>
    \epsf
    -
    \frac{g}{4}
    -
    \left(
        \epsc+\frac{g}{2}
    \right) \\
    &=
    \frac{g}{4}.
\end{split}
\]
Thus $Y$ is more than $g/4$-far from every string of absolute Hamming
weight $W$. After the sample and $W$ have been fixed, the residual cheating strategy is a valid cheating prover for the exact-Hamming-weight subprotocol on the fixed input $Y$. Its acceptance probability is at most $\delta$. Including the probability that $\mathcal{E}$ fails, the total acceptance probability is at most $2\delta$.

The verifier draws exactly $m$ samples. The honest prover queries $f$ once
on each distinct point occurring in the sample and caches the answers.
It therefore makes at most $\min\{m,|\mathcal{X}|\}$ distinct queries to
$f$, and later queries made during the exact-Hamming-weight subprotocol
cause no additional honest-prover queries. The verifier's queries and the
remaining communication are those of that subprotocol at proximity
parameter $g/4$. Sending the ordered sample costs $m\ell$ bits. The stated
bounds follow.
\end{proof}

The next theorem shows that the quadratic sample dependence on $1/g$ is
necessary. The lower bound applies whenever the two Bernoulli parameters
used in the proof remain at least $\alpha$ away from $0$ and $1$.

\begin{theorem}[Black-box sample lower bound]
\label{thm:black-box-sample-lower-bound}
Fix $\delta\in(0,1/4)$ and $\alpha\in(0,1/2)$. Suppose there exists
$\xi\in\{-1,1\}$ such that
$p_{\mathrm{Y}}:=\sigma+\xi\epsc$ and
$p_{\mathrm{N}}:=\sigma+\xi(\epsc+2g)$
both belong to $[\alpha,1-\alpha]$.
Every $(\epsc,\epsf)$-interactive proof with respect to
$\Delta_{\dwHW}$ in the black-box model, with completeness and
soundness errors at most $2\delta$, requires the verifier to draw at
least
\[
    \frac{
        \alpha(1-\alpha)(1-4\delta)^2
    }{
        2g^2
    }
\]
black-box samples in the worst case.

This holds regardless of the number of rounds, the communication, the
number of queries to $f$, and the sample access given to the honest prover.
\end{theorem}

\begin{proof}
Consider the domain $\mathcal{X}=\{0,1\}$ and the fixed public function
$f(a)=a$.
For $p\in[0,1]$, let $D_p=\Ber(p)$. We compare the two instances
$(D_{p_{\mathrm{Y}}},f)$ and $(D_{p_{\mathrm{N}}},f)$. Because
$\wt_{D_p}(f)=p$, the first instance satisfies
$|\wt_{D_{p_{\mathrm{Y}}}}(f)-\sigma|=\epsc$ and is therefore a
completeness instance.

For the second instance,
\[
\begin{split}
    \left|
        \wt_{D_{p_{\mathrm{N}}}}(f)-\sigma
    \right|
    &=
    \epsc+2g \\
    &=
    \epsf+g \\
    &>
    \epsf.
\end{split}
\]
It is therefore a soundness instance. Notice also that
$|p_{\mathrm{N}}-p_{\mathrm{Y}}|=2g$.

Let $\Prover_{\mathrm{Y}}$ be the designated honest-prover strategy on
the first instance. On the second instance, define a cheating prover
$\Prover^*$ that runs $\Prover_{\mathrm{Y}}$ online against the verifier's actual messages. It uses fresh coins distributed as the honest prover's coins and, whenever the simulated honest prover requests a sample, generates an independent sample from $D_{p_{\mathrm{Y}}}$. Thus, conditioned on the transcript so far, $\Prover^*$ has exactly the same response distribution as the honest prover in the completeness execution. This is a valid cheating strategy because soundness quantifies over arbitrary prover strategies.

Suppose the verifier draws at most $s$ black-box samples. For each
execution, draw $s$ independent verifier samples in advance and reveal them in order whenever the verifier requests a sample. This represents the same distribution even when the number and timing of the requests are adaptive. Couple the verifier's coins, the honest prover's coins in the completeness execution and $\Prover^*$'s coins in the soundness execution, and all samples used internally by these two prover strategies. The function oracle is the same in both instances. Consequently, the verifier's complete view in the two executions is obtained by applying the same randomized post-processing to either $D_{p_{\mathrm{Y}}}^{\otimes s}$ or $D_{p_{\mathrm{N}}}^{\otimes s}$. By the data-processing inequality,
\[
    \TV\left(
        \mathsf{View}_{\mathrm{Y}},
        \mathsf{View}_{\mathrm{N}}
    \right)
    \leq
    \TV\left(
        D_{p_{\mathrm{Y}}}^{\otimes s},
        D_{p_{\mathrm{N}}}^{\otimes s}
    \right).
\]

For Bernoulli distributions,
\[
\begin{split}
    \KL\left(
        D_{p_{\mathrm{Y}}}
        \,\middle\|\,
        D_{p_{\mathrm{N}}}
    \right)
    &\leq
    \frac{
        (p_{\mathrm{Y}}-p_{\mathrm{N}})^2
    }{
        p_{\mathrm{N}}(1-p_{\mathrm{N}})
    } \\
    &\leq
    \frac{4g^2}{\alpha(1-\alpha)}.
\end{split}
\]

Tensorization of relative entropy and Pinsker's inequality give
\[
\begin{split}
    \TV\left(
        D_{p_{\mathrm{Y}}}^{\otimes s},
        D_{p_{\mathrm{N}}}^{\otimes s}
    \right)
    &\leq
    \sqrt{
        \frac{1}{2}
        \KL\left(
            D_{p_{\mathrm{Y}}}^{\otimes s}
            \,\middle\|\,
            D_{p_{\mathrm{N}}}^{\otimes s}
        \right)
    } \\
    &\leq
    g
    \sqrt{
        \frac{2s}{\alpha(1-\alpha)}
    }.
\end{split}
\]

Completeness implies that the verifier accepts in the first execution with probability at least $1-2\delta$. Soundness implies that it accepts in the second execution with probability at most $2\delta$. Consequently,
\[
    1-4\delta
    \leq
    \TV\left(
        \mathsf{View}_{\mathrm{Y}},
        \mathsf{View}_{\mathrm{N}}
    \right).
\]

Combining the preceding inequalities gives
\[
    1-4\delta
    \leq
    g
    \sqrt{
        \frac{2s}{\alpha(1-\alpha)}
    }.
\]

Rearranging,
\[
    s
    \geq
    \frac{
        \alpha(1-\alpha)(1-4\delta)^2
    }{
        2g^2
    }.
\]
\end{proof}

\begin{remark}[Boundary regimes]
\label{rem:black-box-boundary}
The interior condition in
Theorem~\ref{thm:black-box-sample-lower-bound} is substantive. For
example, suppose that $\sigma=\epsc=0$, so that $g=\epsf$. A verifier can
draw $s$ independent samples, query $f$ on each sample, and accept if and
only if every answer is zero. This test has perfect completeness. If
$\wt_D(f)>g$, its acceptance probability is at most
\((1-g)^s\leq \exp(-gs)\).

Thus $s=\lceil\ln(1/\delta)/g\rceil$ samples suffice for soundness error at most $\delta$. The quadratic lower bound is therefore not universal near the boundary of $[0,1]$; the optimality result above concerns parameter regimes in which suitable completeness and soundness means remain bounded away from $0$ and $1$.
\end{remark}

\begin{remark}[Dependence on the error probability]
\label{rem:black-box-error-dependence}
Theorem~\ref{thm:black-box-sample-lower-bound} is stated to make the
constant-error dependence explicit. If $\delta$ is allowed to approach
zero, the same pair of instances and the Bretagnolle-Huber inequality give
\[
    s
    \geq
    \frac{\alpha(1-\alpha)}{4g^2}
    \ln\left(\frac{1}{8\delta}\right)
\]
whenever $\delta<1/8$. Thus the $O(\ln(1/\delta)/g^2)$ sample dependence in Theorem~\ref{thm:black-box-whw} is also asymptotically optimal as $\delta$ tends to zero.
\end{remark}

\begin{corollary}
\label{cor:black-box-sample-optimal}
For every fixed $\delta\in(0,1/6]$ and every fixed $\alpha\in(0,1/2)$, the $O(1/g^2)$ verifier-sample complexity of Theorem~\ref{thm:black-box-whw} is optimal up to a constant factor on parameter settings satisfying the interior condition of Theorem~\ref{thm:black-box-sample-lower-bound}.
\end{corollary}

\begin{remark}[Communication in the black-box model]
\label{rem:black-box-communication}
Protocol $\PiWM^{\mathrm{BB}}$ sends the entire realized sample to the
prover and therefore uses $O(m\ell)$ bits for this step. Theorem
\ref{thm:black-box-sample-lower-bound} is only a sample lower bound and
does not imply a corresponding communication lower bound: a different
protocol might avoid revealing the sample explicitly. We leave open whether one can simultaneously obtain $O(1/g^2)$ black-box verifier samples, $O(1/g)$ verifier queries to $f$, and $o(1/g^2)$ communication when the cost of representing domain elements is accounted for separately.
\end{remark}

\subsubsection{DsIPP with Gray-Box Sampling and Query Access}
\label{subsec:whw_gray_box}

Fix the same error parameter $\delta\in(0,1/4)$ as in the preceding
subsection. The ability to choose the sampling oracle's random tape turns the distributional problem into an ordinary Hamming-weight problem over the sampler's random tapes. This is the setting in which the compact structure of Section~\ref{sec:hamming_weight} matters: the sample is now described by a short seed, so the $\Theta(1/g^2)$ prover message of the spot-checking baseline would dominate the communication. Throughout this subsection, the function $f$ and sampling oracle $G$ are fixed before the protocol begins; neither may depend on the verifier's subsequent randomness or queries.

\begin{lemma}[Pullback to the sampler's random tapes]
\label{lem:gray-box-pullback}
Let $G\colon\{0,1\}^{R}\to\mathcal{X}$ satisfy $G(U_R)\sim D$. Define
the length-$2^R$ Boolean string
$z\colon\{0,1\}^{R}\to\{0,1\}$ by $z(r)=f(G(r))$ for every
$r\in\{0,1\}^{R}$. Then $\wt(z)=\wt_D(f)$.
Moreover, one query to $z(r)$ can be simulated using one chosen-randomness call to $G$ and one query to $f$.
\end{lemma}

\begin{proof}
By the definition of relative Hamming weight and the assumption that $G$
pushes the uniform distribution on its random tapes forward to $D$,
\[
\begin{split}
    \wt(z)
    &=
    2^{-R}
    \sum_{r\in\{0,1\}^{R}} z(r) \\
    &=
    2^{-R}
    \sum_{r\in\{0,1\}^{R}} f(G(r)) \\
    &=
    \mathbb{E}_{r\leftarrow U_R}[f(G(r))] \\
    &=
    \mathbb{E}_{X\sim D}[f(X)] \\
    &=
    \wt_D(f).
\end{split}
\]
The query simulation follows directly from the definition
$z(r)=f(G(r))$.
\end{proof}

\begin{protofig}
    {Protocol $\PiWM^{\mathrm{GB}}$}
    {fig:protocol-piwm-gb}

\begin{enumerate}
    \item Define the virtual string
          $z\colon\{0,1\}^{R}\to\{0,1\}$ by $z(r)=f(G(r))$.

    \item The parties run the tolerant logarithmic-round Hamming-weight
          protocol from Section~\ref{sec:hamming_weight}, sequentially
          amplified so that its completeness and soundness errors are at
          most $2\delta$, on oracle access to $z$, with claimed weight
          $\sigma$, completeness radius $\epsc$, and soundness radius
          $\epsf$.

          Whenever either party queries a coordinate $z(r)$, it computes
          that coordinate by invoking $G(r)$ and then querying
          $f(G(r))$.
\end{enumerate}
\end{protofig}

\begin{theorem}[Gray-box upper bound]
\label{thm:gray-box-whw}
Protocol $\PiWM^{\mathrm{GB}}$ in
Figure~\ref{fig:protocol-piwm-gb} has completeness and soundness errors
at most $2\delta$ for distributional weighted Hamming-weight verification
in the gray-box model. In particular, when $\delta\leq 1/6$, it is an
$(\epsc,\epsf)$-interactive proof with respect to
$\Delta_{\dwHW}$ in the sense of
Definition~\ref{def:tolerant-ipp}.

For every fixed $\delta$, using the logarithmic-round protocol from
Section~\ref{sec:hamming_weight}, its resource bounds are as follows:
\begin{itemize}
    \item the verifier makes $O(1/g)$ chosen-randomness calls to $G$ and
          $O(1/g)$ queries to $f$;

    \item the honest prover makes $O(1/g^2)$ chosen-randomness calls to $G$ and
          $O(1/g^2)$ queries to $f$;

    \item the communication is $O(R)+\Otilde(1/g)$ bits, and the number of
          rounds is $O(\log(1/g))$.
\end{itemize}
Any other round-query trade-off established in
Section~\ref{sec:hamming_weight} transfers in the same way.
\end{theorem}

\begin{proof}
Let $z$ be the string defined in
Lemma~\ref{lem:gray-box-pullback}. Since $\wt(z)=\wt_D(f)$, the
distributional promise for $(D,f)$ is exactly the tolerant Hamming-weight
promise for $z$.

Protocol $\PiWM^{\mathrm{GB}}$ simulates the Section~\ref{sec:hamming_weight} protocol on $z$. Whenever that protocol
queries the coordinate indexed by $r\in\{0,1\}^{R}$, the querying party
computes the same answer by invoking $G(r)$ and querying $f(G(r))=z(r)$.
Thus the simulated transcript has exactly the same distribution as a direct execution of the Hamming-weight protocol on $z$. Its completeness and soundness errors are therefore at most $2\delta$.

The virtual string $z$ has length $n=2^R$. The logarithmic-round protocol from Section~\ref{sec:hamming_weight}, with gap $g=\epsf-\epsc$, makes $O(1/g)$ verifier queries and $O(1/g^2)$ honest-prover queries. Each such query becomes one call to $G$ and one query to $f$.

Using the compressed-sample version of that protocol, its communication is $O(\log n)+\Otilde(1/g)$. Since $\log n=R$, this becomes $O(R)+\Otilde(1/g)$. The number of rounds is unchanged and equals $O(\log(1/g))$.
\end{proof}

\subsubsection{Reusable Public Samples}
\label{subsec:reusable-public-samples}

The gray-box model also allows the parties to select a public sample once and reuse it for several later weighted Hamming-weight claims. The following proposition gives the statistical guarantee of this reuse.

\begin{proposition}[Reusable public audit set] 
\label{prop:reusable-audit-set} 
Fix the gray-box sampler $G$ and Boolean functions $f_1,\ldots,f_k\colon\mathcal{X}\to\{0,1\}$ before the public sample is selected, and write $\mu_a:=\wt_D(f_a)$ for $a\in[k]$. Let $r_1,\ldots,r_m$ be independent uniform elements of $\{0,1\}^{R}$, where $\delta_{\rm samp}\in(0,1)$ and 
\[ m \geq \left\lceil \frac{8\ln(2k/\delta_{\rm samp})}{g^2} \right\rceil, \]
and define $Y_i^{(a)}:=f_a(G(r_i))$ for each $a\in[k]$ and $i\in[m]$. Except with probability at most $\delta_{\rm samp}$ over the one-time selection of the public sample, all $k$ strings simultaneously satisfy 
\[ \left| \frac{\lVert Y^{(a)}\rVert_1}{m} - \mu_a \right| \leq \frac{g}{4}. \] 

Conditioned on this event, for any $a\in[k]$ and claimed value $\sigma_a$, the consistency check and exact-Hamming-weight subprotocol used in Protocol~$\PiWM^{\mathrm{BB}}$ can be run on $Y^{(a)}$ without transmitting the sample itself. If that subprotocol is run with proximity parameter $g/4$ and has completeness and soundness errors at most $\delta_{\rm ipp}$, then the resulting execution has completeness and soundness errors at most $\delta_{\rm ipp}$, conditioned on the public sample. Unconditionally, each execution has error at most $\delta_{\rm samp}+\delta_{\rm ipp}$.
\end{proposition}

\begin{proof} 
For each fixed $a\in[k]$, the variables $Y^{(a)}_1,\ldots,Y^{(a)}_m$ are independent Bernoulli random variables with mean $\mu_a$. Hence Hoeffding's inequality gives 
\[ \Pr\left[ \left| \frac{\lVert Y^{(a)}\rVert_1}{m} - \mu_a \right| > \frac{g}{4} \right] \leq 2\exp\left(-\frac{mg^2}{8}\right) \leq \frac{\delta_{\rm samp}}{k}. \] 
A union bound over $a\in[k]$ proves the simultaneous claim. 

Condition on a public sample for which these inequalities hold and fix $a\in[k]$. The completeness and soundness arguments are then exactly the same as in Theorem~\ref{thm:black-box-whw}. For completeness, the honest prover reports $W_a=\lVert Y^{(a)}\rVert_1$, and the concentration bound implies that the consistency check passes. For soundness, any claimed weight $W_a$ that passes the consistency check must differ from the true weight of $Y^{(a)}$ by more than the proximity threshold $g/4$ whenever $|\mu_a-\sigma_a|>\epsf$. Soundness of the exact-Hamming-weight subprotocol therefore applies. Since the public random tapes are already known to both parties, they do not need to be transmitted. Soundness of the exact-Hamming-weight subprotocol holds for every fixed string $Y^{(a)}$ and every cheating prover, even when the public sample is known. Adding the probability that the simultaneous concentration event fails gives the unconditional bound $\delta_{\rm samp}+\delta_{\rm ipp}$. 
\end{proof}

%% file: sections/corrupted_queries.tex
\subsection{Evaluator Disagreement}
\label{sec:corrupted_queries}
The protocols above assume that the prover and the verifier obtain the same evaluation on every sampled point. We now relax this assumption. When an evaluation is supplied by a human, a model owner and an auditor may use
different evaluators whose scores are close on most outputs but unrelated on a small exceptional part of the population. We ask whether the verifier can nevertheless certify the mean computed using the prover's scoring rule.

We model the two scoring rules by fixed evaluation functions. The designated
honest prover queries one function, while the verifier queries the other. Both
functions are fixed before the protocol begins, so the model excludes an
online adversary that chooses an answer after seeing the verifier's query. The
protocol certifies the prover-side mean; the verifier-side function is used
only to check the prover's claim.

The binary-splitting protocol of Aaronson, Gur, Rajgopal, and
Rothblum~\cite{aaronson_et_al:LIPIcs.CCC.2024.24} checks Boolean leaf values.
We extend its leaf check to bounded scores. If the prover claims a score $a$
and the verifier obtains $b$, the verifier records a mismatch with probability
$|a-b|$. This randomized comparison preserves the discrepancy between a false
root claim and the verifier-side sample mean, while charging an honest prover
only for the disagreement already present between the two evaluation
functions.

\subsubsection{Evaluator disagreement model}
\label{subsec:corrupted-query-model}

We state the protocol for finite-precision scores. Fix a public integer
$K\geq 1$, and let $\mathcal{S}_K$ consist of the multiples of $1/K$ in
$[0,1]$. The case $K=1$ gives Boolean evaluations. Let $D$ be a distribution
over a finite domain $\mathcal{X}$, and let
$f_{\Prover},f_{\Verifier}\colon\mathcal{X}\to\mathcal{S}_K$ be the
prover-side and verifier-side evaluation functions. We denote their means by
$\mu_{\Prover}:=\mathbb{E}_{X\sim D}[f_{\Prover}(X)]$ and
$\mu_{\Verifier}:=\mathbb{E}_{X\sim D}[f_{\Verifier}(X)]$.

\begin{definition}[Bounded $(\rho,\gamma)$-mismatch]
\label{def:fixed-score-mismatch}
Let $\rho,\gamma\in[0,1]$. The pair
$(f_{\Prover},f_{\Verifier})$ has bounded $(\rho,\gamma)$-mismatch with
respect to $D$ if there is a set $B\subseteq\mathcal{X}$ such that
$D(B)\leq\rho$ and
$|f_{\Prover}(x)-f_{\Verifier}(x)|\leq\gamma$ for every $x\notin B$.
No restriction is imposed on the two scores on $B$.
\end{definition}

An instance in this section consists of
$(D,f_{\Prover},f_{\Verifier})$ satisfying
Definition~\ref{def:fixed-score-mismatch}. In the black-box model, the
verifier has sampling access to $D$ and query access to $f_{\Verifier}$,
while the designated honest prover has query access to $f_{\Prover}$. A
cheating prover remains unrestricted and may know $D$ and both evaluation
functions completely. The mismatch bound is a promise on the instance; the
protocol does not certify that this promise holds.

For a claim $\sigma\in[0,1]$, define the problem-specific distance by
$\Delta_{\mathrm{mis}}((D,f_{\Prover},f_{\Verifier}),\sigma)
:=|\mu_{\Prover}-\sigma|$. Thus the protocol certifies the prover-side mean,
although the verifier cannot query $f_{\Prover}$.

The mismatch promise yields the following bound.

\begin{lemma}[Mean discrepancy]
\label{lem:fixed-mismatch-mean}
Suppose that $(f_{\Prover},f_{\Verifier})$ has bounded
$(\rho,\gamma)$-mismatch with respect to $D$, and define
$\kappa:=\rho+(1-\rho)\gamma$. Then
\[
    \mathbb{E}_{X\sim D}
    \bigl[|f_{\Prover}(X)-f_{\Verifier}(X)|\bigr]
    \leq \kappa
    \qquad\text{and}\qquad
    |\mu_{\Prover}-\mu_{\Verifier}|\leq\kappa.
\]
\end{lemma}

\begin{proof}
Let $B$ be the exceptional set from
Definition~\ref{def:fixed-score-mismatch}. The pointwise discrepancy is at
most $1$ on $B$ and at most $\gamma$ outside $B$. Therefore
\[
    \mathbb{E}_{X\sim D}
    \bigl[|f_{\Prover}(X)-f_{\Verifier}(X)|\bigr]
    \leq D(B)+(1-D(B))\gamma
    \leq \rho+(1-\rho)\gamma=\kappa.
\]
The mean bound follows from
$|\mathbb{E}[Z]|\leq\mathbb{E}[|Z|]$ applied to
$Z=f_{\Prover}(X)-f_{\Verifier}(X)$.
\end{proof}

Let $0\leq\epsc<\epsf\leq1$ and write $g:=\epsf-\epsc$. In a completeness
instance, the verifier-side mean may be as far as $\epsc+\kappa$ from
$\sigma$. In a soundness instance, it is more than $\epsf-\kappa$ from
$\sigma$. The mean-level separation visible through $f_{\Verifier}$ is
therefore
\(\eta:=(\epsf-\kappa)-(\epsc+\kappa)=g-2\kappa\).

We assume from now on that $\eta>0$. This is precisely the regime in which
Lemma~\ref{lem:fixed-mismatch-mean} leaves a positive separation between the
completeness and soundness ranges visible through $f_{\Verifier}$.

\paragraph{Randomized comparison of two scores.}
For $a,b\in\mathcal{S}_K$, draw $U$ uniformly from $[K]$ and define
\(
    C_K(a,b;U)
    :=
    \mathbf{1}
    \bigl[
        \mathbf{1}[U\leq Ka]
        \neq
        \mathbf{1}[U\leq Kb]
    \bigr].
\)
The two threshold indicators differ for exactly $K|a-b|$ choices of $U$,
and hence
\begin{equation}
\label{eq:randomized-score-comparison}
    \Pr_{U\leftarrow[K]}[C_K(a,b;U)=1]=|a-b|.
\end{equation}
This is a randomized threshold comparison.
For arbitrary real scores in $[0,1]$, the same identity holds with
$U\sim\mathrm{Unif}[0,1]$. We use the finite grid only to implement the
comparison exactly and to count communication bits.

\subsubsection{A robust splitting check}
\label{subsec:robust-splitting-check}

The verifier receives a claimed sum for an $m$-point sample but cannot check
every summand. It instead checks one path in a balanced binary tree of partial
sums. At each node, the prover commits to both child sums before learning
which child the verifier will inspect.

Let $\mathcal{T}_m$ be the balanced binary interval tree over $[m]$. Its root
is $[m]$, and every nonsingleton interval $I$ has consecutive children
$I_0,I_1$ whose sizes differ by at most one. At the beginning of a descent,
the verifier privately samples $J\leftarrow[m]$ and follows the unique path
from the root to $\{J\}$. Conditional on reaching an interval $I$, the next
child is $I_b$ with probability $|I_b|/|I|$.

For an interval $I$, let
$\mathcal{W}_K(I):=\{z/K:z\in\{0,\ldots,K|I|\}\}$. Starting from a root
claim $W\in\mathcal{W}_K([m])$, the prover sends claims for both children
before the verifier reveals which child contains $J$. The verifier checks
that both claims lie in their prescribed ranges and sum to the current claim.
At a leaf $\{J\}$, the remaining claim belongs to $\mathcal{S}_K$.

The following lemma records the property of this check that we need. The flag
$Z=1$ assigned to an inconsistent split defines a relaxed analytical
experiment. The protocol itself will reject immediately when such a split
occurs.

\begin{lemma}[One robust descent]
\label{lem:one-robust-descent}
Fix $v=(v_1,\ldots,v_m)\in\mathcal{S}_K^m$ and a root claim
$W\in\mathcal{W}_K([m])$. Run one descent against an arbitrary prover. If a
range or additivity check fails, set $Z=1$ and end the descent. Otherwise,
upon reaching a leaf $\{J\}$ with claim $a$, draw a fresh $U\leftarrow[K]$
after receiving $a$ and set $Z=C_K(a,v_J;U)$. Conditional on any transcript
fixed before the descent begins,
\begin{equation}
\label{eq:robust-descent-lower-bound}
    \mathbb{E}[Z]
    \geq
    \left|
        \frac{W}{m}-\frac{1}{m}\sum_{j=1}^m v_j
    \right|.
\end{equation}
If the prover fixes $p=(p_1,\ldots,p_m)\in\mathcal{S}_K^m$, sets
$W=\sum_jp_j$, and always sends the true partial sums of $p$, then
\[
    \Pr[Z=1]=\frac{1}{m}\sum_{j=1}^m|p_j-v_j|.
\]
\end{lemma}

\begin{proof}
Fix the prior transcript and the prover's coins used during the descent. We
may therefore analyze a deterministic continuation. For every visited
interval $I$ with claim $c(I)$, define its signed excess by
$e(I):=c(I)-\sum_{j\in I}v_j$. Follow $e(I)/|I|$ until the first
inconsistent split, at which point we stop at the normalized excess of its
parent.

At every consistent split, $e(I)=e(I_0)+e(I_1)$. Since the verifier follows
child $I_b$ with conditional probability $|I_b|/|I|$,
\[
    \mathbb{E}
    \left[
        \left.\frac{e(I_b)}{|I_b|}\,\right| I
    \right]
    =
    \sum_{b\in\{0,1\}}
        \frac{|I_b|}{|I|}\frac{e(I_b)}{|I_b|}
    =
    \frac{e(I)}{|I|}.
\]
Thus the stopped normalized excess is a martingale. Its absolute value is at
most $1$, because every valid claim and the corresponding true partial sum
lie in $[0,|I|]$.

Let $M$ denote its terminal value. Then
$\mathbb{E}[M]=W/m-m^{-1}\sum_jv_j$. If an inconsistent split occurs, then
$Z=1\geq|M|$. Otherwise, the descent reaches a leaf $\{J\}$ with claim $a$,
so $M=a-v_J$, and
Equation~\eqref{eq:randomized-score-comparison} gives
$\mathbb{E}[Z\mid J,a]=|M|$. Hence
\[
    \mathbb{E}[Z]
    \geq \mathbb{E}[|M|]
    \geq |\mathbb{E}[M]|
    =
    \left|
        \frac{W}{m}-\frac{1}{m}\sum_{j=1}^m v_j
    \right|.
\]
Under the honest partial-sum strategy, every split is consistent and $J$ is
uniform over $[m]$. The second claim now follows from
Equation~\eqref{eq:randomized-score-comparison}.
\end{proof}

\subsubsection{Black-box protocol with evaluator disagreement} 
\label{subsec:corrupted-query-protocol}

We now combine the robust descent with concentration over the sampled points.
The verifier first checks that the claimed sample mean lies in the
completeness range. It then repeats the descent enough times to distinguish
the honest disagreement rate, which is close to at most $\kappa$, from the
rate forced by a false root claim, which is close to at least $g-\kappa$.

Fix an error parameter $\delta\in(0,1/3)$ and set
\[
    a:=\frac{\eta}{8},
    \qquad
    m:=\left\lceil\frac{32\ln(12/\delta)}{\eta^2}\right\rceil,
    \qquad
    q:=\left\lceil\frac{24g\ln(12/\delta)}{\eta^2}\right\rceil.
\]

\begin{protofig}
    {Protocol $\Pi_{\mathrm{mis}}^{\mathrm{BB}}$}
    {fig:protocol-corrupted-queries}

\begin{protodesc}
    \item[Common input.]
          A claim $\sigma$, parameters $\epsc,\epsf,\rho,\gamma,K,\delta$,
          and the derived quantities $\kappa=\rho+(1-\rho)\gamma$ and
          $\eta=g-2\kappa>0$.

    \item[Oracle access.]
          The verifier samples from $D$ and queries $f_{\Verifier}$. The
          designated honest prover queries $f_{\Prover}$.
\end{protodesc}

\begin{enumerate}
    \item The verifier draws independent samples $X_1,\ldots,X_m\sim D$ and
          sends the ordered sample $S=(X_1,\ldots,X_m)$ to the prover.

    \item The honest prover sets $p_j:=f_{\Prover}(X_j)$, querying each
          distinct sampled point once, and sends $W:=\sum_{j=1}^m p_j$.

    \item The verifier rejects unless $W\in\mathcal{W}_K([m])$ and
          \[
              \left|\frac{W}{m}-\sigma\right|\leq\epsc+a.
          \]

    \item The parties perform $q$ sequential descents of $\mathcal{T}_m$,
          each rooted at $W$. In every descent, the prover sends both child
          claims before the verifier reveals the child containing its private
          index $J_t$. If a range or additivity check fails, the verifier
          rejects immediately. Otherwise, the descent reaches $\{J_t\}$ with
          claim $A_t$. The verifier queries
          $V_t:=f_{\Verifier}(X_{J_t})$, draws a fresh
          $U_t\leftarrow[K]$ after receiving $A_t$, and sets
          $Z_t:=C_K(A_t,V_t;U_t)$.

    \item If no previous check failed, the verifier accepts if and only if
          \[
              \sum_{t=1}^q Z_t\leq\frac{gq}{2}.
          \]
\end{enumerate}
\end{protofig}

\begin{theorem}[Black-box upper bound with evaluator disagreement] 
\label{thm:corrupted-query-upper-bound}
Let $0\leq\epsc<\epsf\leq1$, write $g=\epsf-\epsc$, and suppose that
$(f_{\Prover},f_{\Verifier})$ has bounded $(\rho,\gamma)$-mismatch with
respect to $D$. Define $\kappa:=\rho+(1-\rho)\gamma$ and
$\eta:=g-2\kappa$. If $\eta>0$, Protocol
$\Pi_{\mathrm{mis}}^{\mathrm{BB}}$ is an
$(\epsc,\epsf)$-interactive proof with respect to
$\Delta_{\mathrm{mis}}$, with completeness and soundness errors at most
$\delta$.

The verifier draws
$m=O(\ln(1/\delta)/\eta^2)$ samples from $D$ and makes
$q=O(g\ln(1/\delta)/\eta^2)$ queries to $f_{\Verifier}$. The honest prover
draws no samples and makes at most $\min\{m,|\mathcal{X}|\}$ distinct queries
to $f_{\Prover}$. If a point of $\mathcal{X}$ is encoded with $\ell$ bits,
the communication is
\[
    O\bigl(m\ell+q\log m\log(mK)\bigr)
\]
bits, and the number of rounds is $O(q\log m)$.
\end{theorem}

\begin{proof}
For the sample drawn in the first step, define
\[
    \overline p
    :=\frac{1}{m}\sum_{j=1}^m f_{\Prover}(X_j),
    \qquad
    \overline v
    :=\frac{1}{m}\sum_{j=1}^m f_{\Verifier}(X_j),
    \qquad
    \overline d
    :=\frac{1}{m}\sum_{j=1}^m
       |f_{\Prover}(X_j)-f_{\Verifier}(X_j)|.
\]

\emph{Completeness.}
Suppose that $|\mu_{\Prover}-\sigma|\leq\epsc$ and that the prover follows
the designated strategy. Hoeffding's inequality and
Lemma~\ref{lem:fixed-mismatch-mean} give
\[
    \Pr[|\overline p-\mu_{\Prover}|>a]
    \leq 2e^{-2ma^2}\leq\frac{\delta}{6},
    \qquad
    \Pr[\overline d>\kappa+a]
    \leq e^{-2ma^2}\leq\frac{\delta}{12}.
\]
Condition on the complementary events. Since $W/m=\overline p$, the root
check passes. In addition,
$\overline d\leq\kappa+\eta/8=g/2-3\eta/8$.

The honest prover supplies the true partial sums, so every consistency check
passes. Conditional on the sample,
Lemma~\ref{lem:one-robust-descent} shows that
$Z_1,\ldots,Z_q$ are independent Bernoulli variables with mean
$\overline d$. Bernstein's inequality gives
\[
    \Pr\left[\sum_{t=1}^q Z_t>\frac{gq}{2}\right]
    \leq
    \exp\left(-\frac{9q\eta^2}{64g}\right)
    \leq\frac{\delta}{12},
\]
where we used $2\overline d+\eta/4\leq g$ and the definition of $q$.
The total rejection probability is therefore at most $\delta$.

\emph{Soundness.}
Suppose that $|\mu_{\Prover}-\sigma|>\epsf$, and fix an arbitrary cheating
prover. Hoeffding's inequality gives
\[
    \Pr[|\overline v-\mu_{\Verifier}|>a]
    \leq 2e^{-2ma^2}
    \leq\frac{\delta}{6}.
\]
Fix a sample in the complementary event and a root claim $W$ that passes the
root check. Lemma~\ref{lem:fixed-mismatch-mean} and the triangle inequality
imply
\begin{equation}
\label{eq:false-root-verifier-gap}
    \left|\frac{W}{m}-\overline v\right|
    >
    \epsf-\kappa-a-(\epsc+a)
    =
    \frac{g}{2}+\frac{\eta}{4}.
\end{equation}

For the analysis, consider the relaxed verifier that records $Z_t=1$ and
continues whenever a consistency check fails. This can only increase the
acceptance probability. Condition on the complete transcript before descent
$t$. The root claim remains $W$, so
Lemma~\ref{lem:one-robust-descent} and
Equation~\eqref{eq:false-root-verifier-gap} give
\[
    \Pr[Z_t=1\mid\text{prior transcript}]
    \geq
    p_*:=\frac{g}{2}+\frac{\eta}{4}.
\]
This bound allows the prover to adapt its claims between descents. A sequence
of Bernoulli variables whose conditional success probabilities are at least
$p_*$ stochastically dominates $\operatorname{Bin}(q,p_*)$. Since
$p_*\leq3g/4$, a Chernoff bound yields
\[
    \Pr\left[\sum_{t=1}^qZ_t\leq\frac{gq}{2}\right]
    \leq
    \exp\left(-\frac{q\eta^2}{24g}\right)
    \leq\frac{\delta}{12}.
\]
After accounting for the sampling event, the cheating prover is accepted
with probability at most $\delta$.

The verifier makes at most one query to $f_{\Verifier}$ per descent. The
honest prover evaluates $f_{\Prover}$ once on each distinct sampled point and
reuses the stored values in every descent. Sending the ordered sample costs
$m\ell$ bits. Each descent has depth $O(\log m)$ and exchanges
$O(\log(mK))$ bits per level, which proves the remaining bounds.
\end{proof}

The following regime recovers the query exponents of the protocol with a
common evaluation function.

\begin{corollary}[Mismatch below half the tolerance gap]
\label{cor:corrupted-query-fixed-fraction}
Fix a constant $c<1/2$ and suppose that $\kappa\leq cg$. For every fixed
error probability, Protocol $\Pi_{\mathrm{mis}}^{\mathrm{BB}}$ uses
$O(1/g^2)$ honest-prover evaluations and $O(1/g)$ verifier evaluations. The
hidden constants are proportional to $(1-2c)^{-2}$.
\end{corollary}

\begin{proof}
The assumption gives $\eta\geq(1-2c)g$. Substituting this bound into
Theorem~\ref{thm:corrupted-query-upper-bound} proves the claim.
\end{proof}

\begin{remark}[A banded variant with fewer verifier queries]
\label{rem:banded-variant}
The randomized comparison $C_K$ at the leaves of
Figure~\ref{fig:protocol-corrupted-queries} is triggered with probability $|A_t-V_t|$. Because of that, the honest prover ``pays'' for the in-band disagreement $\gamma$ on every leaf, not only on the exceptional set. This is what causes the threshold $gq/2$ and makes soundness an estimation problem
for an event of rate $\Theta(g)$ at precision $\Theta(\eta)$, hence the factor $g/\eta^2$. In Appendix~\ref{app:banded-protocol}, we describe a simple variant, Protocol
$\Pi_{\mathrm{band}}^{\mathrm{BB}}$ (Figure~\ref{fig:protocol-band}), which replaces the comparison by the deterministic band test: the verifier records $Z_t:=\mathbf{1}[|A_t-V_t|>\gamma]$ and accepts if and only if $\sum_tZ_t\leq(\rho+\eta/2)q_{\mathrm{b}}$. The variant achieves a tighter bound: an honest prover now fails the test only at leaves that fall in the exceptional set, at rate at most $\rho$, while a false root claim still forces a failure rate of at least $\rho+\Omega(\eta)$. The verifier
therefore estimates an event of rate $\Theta(\rho+\eta)$ instead of $\Theta(g)$, and
\[
q_{\mathrm{b}}=O \left(\frac{(\rho+\eta)\ln(1/\delta)}{\eta^2}\right)
\]
descents suffice, with the same $m$, the same honest-prover complexity, and the same gray-box implementation (Theorem~\ref{thm:band-upper-bound}).

Since $g=2\rho+2(1-\rho)\gamma+\eta$, the saving is a factor of order $1+(\rho+2\gamma)/(\rho+\eta)$. It is a constant for Boolean scores ($\gamma=0$), where the protocol above is already optimal (see Corollary~\ref{cor:query-landscape}), and it becomes large exactly when the in-band disagreement $\gamma$ dominates $\rho+\eta$, that is, when the two evaluators are close everywhere but not identical
and the claim is tight. For instance, with $\rho=0$, $\gamma=0.1$, and
$\eta=0.01$ (so $g\approx0.21$), the leading factor $g/\eta^2\approx2100$
becomes $(\rho+\eta)/\eta^2=100$: about twenty times fewer verifier
evaluations, up to the constants of the two analyses.

We keep the protocol above in the main body of the paper since it comes with a much simpler and shorter proof: the analysis of the variant is more involved because the band test does not
have the exact unbiasedness of Lemma~\ref{lem:one-robust-descent} and one must account for the in-band shaving available to a cheating prover. We state the variant and its analysis in Appendix~\ref{app:banded-protocol} for completeness. We also provide in Appendix~\ref{app:optimality-disagreement} a lower bound on the resources of the variant and show that it is essentially optimal (Corollary~\ref{cor:query-landscape}).
\end{remark}

\subsubsection{Gray-box protocol with evaluator disagreement}
\label{subsec:corrupted-query-gray-box}

The same protocol applies in the gray-box model of
Definition~\ref{def:gray-box-sampling}. Let
$G\colon\{0,1\}^R\to\mathcal{X}$ satisfy $G(U_R)\sim D$, and define
$z_{\Prover}(r):=f_{\Prover}(G(r))$ and
$z_{\Verifier}(r):=f_{\Verifier}(G(r))$. If $B$ is the exceptional set from
Definition~\ref{def:fixed-score-mismatch}, then its preimage under $G$ has
uniform measure $D(B)\leq\rho$. Hence
$(z_{\Prover},z_{\Verifier})$ satisfies the same bounded mismatch promise on
the random-tape domain.

For constant error, the verifier can describe the virtual sample using a
pairwise-independent family. Let
$s:=\max\{R,\lceil\log m\rceil\}$. The verifier chooses
$\alpha,\beta\leftarrow\mathbb{F}_{2^s}$ uniformly and sends them to the
prover. Fix distinct field elements $j_1,\ldots,j_m$, set
$H(j):=\alpha j+\beta$, and let $r_i$ be the first $R$ coordinates of
$H(j_i)$ under a fixed $\mathbb{F}_2$-linear identification of
$\mathbb{F}_{2^s}$ with $\{0,1\}^s$. The tapes
$r_1,\ldots,r_m$ are uniform and pairwise independent, and the seed
$(\alpha,\beta)$ uses $O(R+\log m)$ bits.

The gray-box protocol $\Pi_{\mathrm{mis}}^{\mathrm{GB}}$ replaces the first
step of Figure~\ref{fig:protocol-corrupted-queries} by this seed generation
and sets $X_j:=G(r_j)$. The verifier invokes $G$ only at the leaves it checks,
while the honest prover invokes $G$ on the entire virtual sample.

\begin{corollary}[Evaluator disagreement in the gray-box model]
\label{cor:corrupted-query-gray-box}
For every fixed error probability, the conclusion of
Theorem~\ref{thm:corrupted-query-upper-bound} holds in the gray-box model
with $m=O(1/\eta^2)$ and $q=O(g/\eta^2)$. The verifier makes at most $q$
chosen-randomness calls to $G$ and $q$ queries to $f_{\Verifier}$, while the
honest prover makes at most $m$ calls to $G$ and $m$ queries to
$f_{\Prover}$. The communication is
\[
    O\bigl(R+\log m+q\log m\log(mK)\bigr)
\]
bits, and the number of rounds is $O(q\log m)$.

In particular, if $\kappa\leq cg$ for a fixed $c<1/2$, the verifier uses
$O(1/g)$ calls and evaluations, while the honest prover uses $O(1/g^2)$
calls and evaluations.
\end{corollary}

\begin{proof}
For every fixed $h\colon\{0,1\}^R\to[0,1]$, pairwise independence of the
tapes gives variance at most $1/(4m)$ for the empirical mean of
$h(r_1),\ldots,h(r_m)$. Apply Chebyshev's inequality to the prover scores,
the verifier scores, and their pointwise discrepancies. For every fixed
error probability, choosing $m=O(1/\eta^2)$ with a sufficiently large
constant makes the three sample events used in the proof of
Theorem~\ref{thm:corrupted-query-upper-bound} hold with the required
probability. Conditional on those events, the descent analysis is unchanged.
The remaining bounds follow from the seed length and the resource accounting
in Theorem~\ref{thm:corrupted-query-upper-bound}.
\end{proof}

%% file: sections/applications.tex
We now describe the auditing application that motivates the distinction
between black-box and gray-box access. A model owner wishes to certify a
public statistical claim about a fixed model. Producing one evaluated audit record may be expensive: it may require generating a synthetic input, running the model, and obtaining a human or otherwise costly evaluation of the output. The protocols above allow the model owner to perform the larger number of evaluations needed to support the claim, while an auditor checks only a much smaller number of them.

\subsection{Auditing a model on generated data}
\label{sec:generated-audits}

Let
\(
    G\colon\{0,1\}^{R}\longrightarrow\mathcal{X}
\)
be a deterministic generator for audit instances, and let $M\colon\mathcal{X}\to\mathcal{Y}$ be the model being audited. An audit
criterion is specified by a Boolean evaluation rule
\(
    \phi\colon\mathcal{X}\times\mathcal{Y}\longrightarrow\{0,1\}.
\)
The rule may also use labels or other metadata included in the generated
record. Define
\[
    f_{\phi, M}(x):=\phi(x,M(x))
\]
\[
    z_{\phi}(r)
    :=
    f_{\phi}(G(r))
    =
    \phi\bigl(G(r),M(G(r))\bigr),
    \qquad r\in\{0,1\}^{R}.
\]

If $D_G$ is the distribution of $G(U_R)$, then
\(
    \wt(z_{\phi})
    =
    \mathbb{E}_{X\sim D_G}
    \bigl[\phi(X,M(X))\bigr]
    =
    \wt_{D_G}(f_{\phi}).
\)
Consequently, a claim about the average value of $\phi$ is exactly a
distributional weighted Hamming-weight claim. The gray-box protocol of
Theorem~\ref{thm:gray-box-whw} verifies such a claim using $O(1/g)$ calls to $G$ and $O(1/g)$ evaluations of $f_{\phi}$ by the auditor. The honest model owner performs $O(1/g^2)$ such evaluations.

\subsection{A reusable public audit set}
\label{sec:public-audit-set}

The honest model owner's evaluations can be reused when several auditors are to verify claims about the same fixed model and audit distribution. The audit set may be selected by a trusted party. Alternatively, the parties may use a \emph{nothing-up-my-sleeve} procedure: they first fix the generator, model, evaluation rules, audit parameters, and the deterministic rule for deriving random tapes. A future public value, such as the outcome of a designated lottery, then determines the random tapes. The resulting audit set is public and reproducible.

The statistical guarantee underlying this reuse is given by Proposition~\ref{prop:reusable-audit-set}. To apply it here, for each evaluation rule $\phi_a$ define \(f_a(x) := \phi_a(x,M(x))\). Then
\(
    \wt_{D_G}(f_a)
    =
    \mathbb{E}_{X\sim D_G}
    \bigl[\phi_a(X,M(X))\bigr].
\)
Thus a single public sample of
\[
    m
    =
    O\left(
        \frac{\log(k/\delta_{\rm samp})}{g^2}
    \right)
\]
random tapes simultaneously approximates the population means of all $k$ fixed evaluation rules, except with probability
$\delta_{\rm samp}$.

The model owner runs $M$ on all $m=O(\log(k/\delta_{\rm samp})/g^2)$ generated inputs and may cache the resulting outputs. Each auditor uses only the verifier queries of the exact-Hamming-weight protocol, namely $O(1/g)$ evaluations in the logarithmic-round instantiation from
Section~\ref{sec:hamming_weight}. Because every party can reconstruct the
ordered audit set from the public value, the $O(mR)$ bits that would be needed to transmit all random tapes are avoided. The remaining communication is that of the exact-weight subprotocol.

The same model outputs can support several fixed rules
$\phi_1,\ldots,\phi_k$. A rule requiring a new human judgment or another
external score still incurs that additional evaluation cost. Likewise, the soundness error of the interactive proof applies separately to each
execution. If a single guarantee is required across $t$ executions, their
proof errors must be reduced accordingly, for example by setting the error of each execution to at most $\delta_{\rm ipp}/t$.

The public value selects the audit set; it does not replace the fresh
randomness required inside the interactive proof. In particular, the model owner must not learn the auditor's later challenges before sending the messages that those challenges are meant to test.

The formal statement idealizes the selected tapes as independent and
uniform. A concrete public source must supply sufficient unpredictability
and must be fixed in advance together with an unambiguous tape-derivation
rule. Here ``unbiasable'' has an operational and economic meaning, for example an audit participant should have no feasible and worthwhile way to steer the public outcome after the audited objects have been fixed. A lottery is a natural candidate precisely because a party capable of steering its result would normally have a direct financial use for that capability.

\subsection{Black-box and gray-box audits}
\label{sec:audit-access-comparison}

The two access models from Section~\ref{sec:distribution_weighted_hamming_weight} lead to different audit costs. With only black-box access to a real population, Protocol~$\PiWM^{\mathrm{BB}}$ requires \(O(1/g^2)\) fresh population samples.

The verifier evaluates the model or audit rule on only $O(1/g)$ of these records, but it must obtain the whole sample and send it to the model owner. In the interior regimes of
Theorem~\ref{thm:black-box-sample-lower-bound}, the quadratic number of
population samples is necessary. 
With chosen-randomness access to a generator, the parties can instead address the coordinates $z_{\phi}(r)=\phi(G(r),M(G(r)))$ directly. Theorem~\ref{thm:gray-box-whw} then reduces the auditor's use of the generator from $O(1/g^2)$ black-box samples to $O(1/g)$ chosen-randomness calls. Proposition~\ref{prop:reusable-audit-set} gives a complementary deployment where the parties publicly select one sample of size $O(1/g^2)$, the model owner evaluates it once, and later auditors verify its claimed statistics with $O(1/g)$ evaluations each. The first option avoids materializing a shared audit set; the second amortizes the model owner's work across auditors and across fixed audit criteria.

\subsection{Audit criteria as weighted Hamming-weight claims}
\label{sec:audit-criteria}

We next instantiate the Boolean rule $\phi$ for the criteria discussed above. Let an audit record contain an input $X$, any required metadata, and the model output. For an event $E$ determined by this record, write
\[
    p_E:=\Pr[E].
\]
Then $p_E$ is the weighted Hamming weight of the indicator $\mathbf{1}[E]$.

Conditional rates also reduce to weighted Hamming weights. If $p_C>0$,
then
\[
    \Pr[E\mid C]
    =
    \frac{p_{E\cap C}}{p_C}.
\]
Thus a comparison
\(
    \left|
        \Pr[E_0\mid C_0]
        -
        \Pr[E_1\mid C_1]
    \right|
    \leq \tau
\)
can be checked by certifying the four weights $p_{E_0\cap C_0},p_{C_0},p_{E_1\cap C_1},p_{C_1}$ and testing
\[
    \left|
        p_{E_0\cap C_0}p_{C_1}
        -
        p_{E_1\cap C_1}p_{C_0}
    \right|
    \leq
    \tau p_{C_0}p_{C_1}.
\]
The comparison must include slack for the additive tolerances of the four
weight claims. If either conditioning event has small probability, those
tolerances must be correspondingly smaller. Thus conditional auditing of a rare group is more expensive unless $G$ directly generates records from the relevant conditional distribution.

\paragraph{Accuracy.}
Suppose an audit record contains a trustworthy outcome label $Y$, and let
$\widehat{Y}=M(X)$ be the model's prediction. Its error rate is
\[
    \Pr[\widehat{Y}\neq Y]
    =
    \wt_D\bigl(\mathbf{1}[\widehat{Y}\neq Y]\bigr).
\]
Accuracy is the complementary weight. Either claim therefore requires one weighted Hamming-weight verification.

\paragraph{Group fairness.}
Let $A\in\{0,1\}$ be a protected attribute. Statistical parity compares
\[
    \Pr[\widehat{Y}=1\mid A=0]
    \qquad\text{and}\qquad
    \Pr[\widehat{Y}=1\mid A=1].
\]
It reduces to the four weights of the events $A=a$ and $A=a\wedge\widehat{Y}=1$, for $a\in\{0,1\}$.
Equal opportunity compares
\[
    \Pr[\widehat{Y}=1\mid A=0,Y=1]
    \qquad\text{and}\qquad
    \Pr[\widehat{Y}=1\mid A=1,Y=1],
\]
and therefore uses the four weights obtained by replacing the conditioning events with $A=a\wedge Y=1$. Equalized odds makes the analogous comparison for both $Y=1$ and $Y=0$ \cite{3157382.3157469}. It uses the eight weights of
\[
    A=a\wedge Y=y
    \qquad\text{and}\qquad
    A=a\wedge Y=y\wedge\widehat{Y}=1,
    \qquad a,y\in\{0,1\}.
\]

\paragraph{Harmlessness and usefulness.}
Let $h(W)$ indicate that the model output in an audit record $W$ is harmful, and let $u(W)$ indicate that it is useful. The corresponding population rates are $\wt_D(h)$ and $\wt_D(u)$, so each is one weighted
Hamming-weight claim. The evaluation procedure must be fixed before the
audit set is selected. If the model owner and auditor can disagree on these judgments, the resulting two-oracle issue is the subject of Section~\ref{sec:corrupted_queries}.

The same reduction handles bounded finite-precision ratings. Suppose
$s(W)\in\{0,1/K,\ldots,1\}$ and define
\[
    b_s(W,j)
    :=
    \mathbf{1}[j\leq Ks(W)],
    \qquad j\in[K].
\]
Under the product of $D$ and the uniform distribution on $[K]$,
\(
    \wt_{D\times U_{[K]}}(b_s)
    =
    \mathbb{E}_{W\sim D}[s(W)].
\)
Hence an average bounded score, including an average harmlessness or usefulness rating, is a weighted Hamming weight over an enlarged domain.

\paragraph{Calibration.}
Suppose the model reports a score $S$ in a finite set $\mathcal{T}\subseteq[0,1]$. Calibration at $t\in\mathcal{T}$ asks that
\[
    \Pr[Y=1\mid S=t]=t.
\]
Provided $\Pr[S=t]>0$, this is equivalent to
\(
    p_{Y=1\wedge S=t}
    =
    t\,p_{S=t}.
\)
Calibration at one score value therefore reduces to two weighted Hamming-weight claims and a deterministic comparison. For a public score bin $B$, binned calibration compares
\[
    \mathbb{E}\bigl[Y\mathbf{1}[S\in B]\bigr]
    \qquad\text{and}\qquad
    \mathbb{E}\bigl[S\mathbf{1}[S\in B]\bigr].
\]
The first quantity is a Boolean mean, while the second reduces to a Boolean mean through the finite-precision lifting above. Repeating the comparison over a fixed collection of bins certifies binned calibration. For an additive conditional calibration tolerance, the permissible error in these non-normalized means scales with $\Pr[S\in B]$.

\paragraph{Robustness.}
Let the audit generator output a pair $(X,\widetilde{X})$, where $\widetilde{X}$ is a permitted perturbation of $X$. The Boolean rule
\[
    \phi_{\rm rob}(X,\widetilde{X})
    :=
    \mathbf{1}\bigl[M(X)\neq M(\widetilde{X})\bigr]
\]
records a failure of prediction invariance. Its weighted Hamming weight is the average failure probability under the perturbation distribution chosen by the generator. Other fixed Boolean failure rules can be treated in the same way. This application certifies average-case robustness under that distribution; it does not certify robustness against every permitted perturbation.

\subsubsection{Normative Discussion about the Distance Notion}

In Section~\ref{sec:prelims}, we define tolerant $\IPP$s with a general notion of distance. This generalization becomes particularly important when moving beyond the simple weighted Hamming Weight problem to more complex audit criteria commonly used in machine learning. Indeed, some of these criteria admit multiple distance metrics, and these metrics need not coincide with the minimum total variation distance between $D_f$ and the set of distributions satisfying the criterion.

As an example, the distance to statistical parity, i.e., for $R$ the output of the model and $A$ the group attribute $\Pr[R=1\mid A=0]=\Pr[R=1\mid A=1]$, is either measured through the statistical parity difference 
\[
    \Delta_\mathrm{Diff}(D_f, \Pi_\mathrm{StatParity})=|\Pr[R=1\mid A=0]-\Pr[R=1\mid A=1]|
\]
or the statistical parity ratio 
\[
   \Delta_\mathrm{Ratio}(D_f, \Pi_\mathrm{StatParity})= 1-\frac{\min(\Pr[R=1\mid A=0],\Pr[R=1\mid A=1])}{\max(\Pr[R=1\mid A=0],\Pr[R=1\mid A=1])}.
\]
However, considering that the labeling function is $f(X)=(M(X),A_X)$ with $M$ the model and $A_X$ the group attribute of $X$ (i.e., $f(X)=(R,A)$), two distributions $D_f$ and $D'_f$ can be at the same distance according to $\Delta_\mathrm{Diff}$ or $\Delta_\mathrm{Ratio}$ but at a different total variation distance from the closest distribution in $\Pi_\mathrm{StatParity}$\footnote{For example, for $D_f$ and $D'_f$ such that $D_f(0,0)=D_f(1,0)=D_f(0,1)=1/5$, $D_f(1,1)=2/5$, and $D'_f(0,0)=D'_f(1,0)=1/10$, $D'_f(0,1)=4/15$, and $D'_f(1,1)=8/15$ then $\Delta_\mathrm{Diff}(D_f, \Pi_\mathrm{StatParity})=\Delta_\mathrm{Diff}(D'_f, \Pi_\mathrm{StatParity})$ but $\Delta_\mathrm{TV}(D_f, \Pi_\mathrm{StatParity})\neq\Delta_\mathrm{TV}(D'_f, \Pi_\mathrm{StatParity})$.}.

Indeed, choosing which distance measure is appropriate in a given context is ultimately an important normative discussion in machine learning evaluation. This normative discussion about defining a distance to a criterion can be approached both from a theoretical computer science perspective (e.g., how to measure the distance to (multi-)calibrate models \cite{ITCS:DDHLS26, STOC:BGHN23}) and from a philosophical perspective (e.g., what it means to be fair when diverging from the ideal of a distributive justice criterion \cite{Hertweck_Heitz_Loi_2024}). We therefore leave the choice of distance notion unspecified in our framework, as its appropriateness depends on the particular evaluation context.

\subsection{Scope of the certificate}
\label{sec:audit-certificate-scope}

A gray-box certificate concerns the distribution $D_G$ generated by $G$.
It does not, by itself, establish the same claim for a different real
population. If one has an independently justified bound
\[
    \TV(D_{\rm real},D_G)\leq\tau,
\]
then every Boolean evaluation rule satisfies
\[
    \left|
        \wt_{D_{\rm real}}(f_{\phi})
        -
        \wt_{D_G}(f_{\phi})
    \right|
    \leq\tau.
\]
Thus a certificate establishing
$|\wt_{D_G}(f_{\phi})-\sigma|\leq\epsc$ implies
\[
    |\wt_{D_{\rm real}}(f_{\phi})-\sigma|
    \leq \epsc+\tau.
\]
Our protocols do not provide such a bound on
$\TV(D_{\rm real},D_G)$. Establishing that a generated distribution is
close to the real population is a separate statistical problem and, in
general, may itself require substantial sample complexity or additional
structural assumptions. Without an independently justified fidelity
bound, the certificate therefore applies only to the generated audit
distribution.

Conditional criteria are more sensitive to distributional mismatch. Let
$E$ and $C$ be events, and suppose that
\[
    \Pr_{D_G}[C]\ge \beta>\tau.
\]
Then $\Pr_{D_{\rm real}}[C]\ge\beta-\tau$, and
\[
    \left|
        \Pr_{D_{\rm real}}[E\mid C]
        -
        \Pr_{D_G}[E\mid C]
    \right|
    \le
    \frac{2\tau}{\beta-\tau}.
\]
Thus even a small total-variation error may substantially affect a
conditional claim when the conditioning event has small probability.
This qualification is especially relevant for criteria involving small
groups, since synthetic data may represent minority and low-density
regions poorly~\cite{3618408.3619856}.

Finally, the proof certifies the numerical claim defined by $\phi$. Its
interpretation still depends on the audit record and evaluation rule.
Accuracy and group-fairness claims require reliable labels; harmlessness
and usefulness require a fixed judgment procedure; calibration depends
on the chosen score values or bins; and robustness depends on the
perturbation distribution. These choices must be fixed and stated as
part of the audit claim.

The certificate also applies only to the fixed generator, model, and
evaluation rule used in the audit. Updating $G$, changing $M$, or
modifying $\phi$ produces a different audit claim and requires a new
certificate. When the model owner and auditor may disagree on the value
of $\phi$, the guarantee must instead be interpreted through the
evaluator-disagreement model of Section~\ref{sec:corrupted_queries}.

%% file: sections/disclosure.tex
No generative AI tool was involved in finding, stating, or proving any of the results in the main body of this paper. ChatGPT and Claude were used to check for grammar, identify typos, and assess the correctness of proofs or calculations. The banded protocol in Appendix~\ref{app:banded-protocol} was also found by the authors without AI. However, AI tools (Claude Fable 5.1) were used substantively in extending the analysis of Protocol $\Pi_{\mathrm{mis}}^{\mathrm{BB}}$ to the banded variant in Appendix~\ref{app:banded-protocol} and in providing tight and concise proofs for the lower bounds in Appendix~\ref{app:optimality-disagreement}. The lower bounds and their proofs are fairly standard and not central to the contribution; we include them in an appendix for completeness. The authors verified and refined the AI-assisted proofs in the appendix.

%% file: sections/appendix/banded_protocol.tex
This appendix presents and analyzes Protocol
$\Pi_{\mathrm{band}}^{\mathrm{BB}}$, the variant of Protocol
$\Pi_{\mathrm{mis}}^{\mathrm{BB}}$ (Figure~\ref{fig:protocol-corrupted-queries})
described in Remark~\ref{rem:banded-variant}. It differs from
$\Pi_{\mathrm{mis}}^{\mathrm{BB}}$ only in the test performed at the
leaves of the descents and in the acceptance threshold, and it reduces the
verifier query complexity from $O(g\ln(1/\delta)/\eta^2)$ to
$O((\rho+\eta)\ln(1/\delta)/\eta^2)$ (this never
exceeds the query complexity of Theorem~\ref{thm:corrupted-query-upper-bound} and it is smaller by a
factor of order $g/(\rho+\eta)$ whenever the in-band tolerance
$(1-\rho)\gamma$ dominates $\rho+\eta$). We first state the banded analogue of Lemma~\ref{lem:one-robust-descent}:

\begin{lemma}[One banded descent]
\label{lem:banded-descent}
Let $\gamma\in[0,1)$. Fix $v=(v_1,\ldots,v_m)\in\mathcal{S}_K^m$ and a
root claim $W\in\mathcal{W}_K([m])$. Run one descent of
$\mathcal{T}_m$ against an arbitrary prover. If a range or additivity
check fails, set $Z=1$ and end the descent. Otherwise, upon reaching a
leaf $\{J\}$ with claim $a$, set
$Z:=\mathbf{1}\bigl[|a-v_J|>\gamma\bigr]$. Conditional on any transcript
fixed before the descent begins,
\[
    \mathbb{E}[Z]
    \geq
    \frac{1}{1-\gamma}
    \left(
        \left|\frac{W}{m}-\frac{1}{m}\sum_{j=1}^mv_j\right|
        -\gamma
    \right)^{\!+}.
\]
If the prover fixes $p\in\mathcal{S}_K^m$, sets $W=\sum_jp_j$, and always
sends the true partial sums of $p$, then
\[
    \Pr[Z=1]
    =
    \frac{1}{m}\,
    \bigl|\{j\in[m]:|p_j-v_j|>\gamma\}\bigr| .
\]
\end{lemma}

\begin{proof}
Let $M$ be the terminal value of the ``stopped normalized excess'' martingale defined in the proof of Lemma~\ref{lem:one-robust-descent}, so that
$|M|\leq1$ and
$\mathbb{E}[M]=W/m-m^{-1}\sum_jv_j$. If the descent is stopped by a failure (that is, a range or additive check fails during the descent), then
$Z=1\geq(|M|-\gamma)^+/(1-\gamma)$. Otherwise $M=a-v_J$, and
$(|M|-\gamma)^+\leq(1-\gamma)Z$ because $|M|\leq1$. In both cases
\[
    \mathbb{E}[Z]
    \geq
    \frac{\mathbb{E}\bigl[(|M|-\gamma)^+\bigr]}{1-\gamma}
    \geq
    \frac{\bigl(\mathbb{E}[|M|]-\gamma\bigr)^{+}}{1-\gamma}
    \geq
    \frac{\bigl(|\mathbb{E}[M]|-\gamma\bigr)^{+}}{1-\gamma},
\]
where the middle step uses Jensen's inequality for the convex function
$x\mapsto(x-\gamma)^+$. Under the honest strategy, every split is
consistent, $J$ is uniform over $[m]$, and the leaf claim is exactly
$p_J$, which gives the second statement.
\end{proof}

Fix an error parameter $\delta\in(0,1/3)$ and set
\[
    a:=\frac{\eta}{8},
    \qquad
    m:=\left\lceil\frac{32\ln(12/\delta)}{\eta^2}\right\rceil,
    \qquad
    q_{\mathrm{b}}
    :=\left\lceil\frac{64(\rho+\eta)\ln(12/\delta)}{\eta^2}\right\rceil.
\]

\begin{protofig}
    {Protocol $\Pi_{\mathrm{band}}^{\mathrm{BB}}$}
    {fig:protocol-band}

\begin{protodesc}
    \item[Common input.]
          A claim $\sigma$, parameters $\epsc,\epsf,\rho,\gamma,K,\delta$,
          and the derived quantities $\kappa=\rho+(1-\rho)\gamma$ and
          $\eta=g-2\kappa>0$.

    \item[Oracle access.]
          The verifier samples from $D$ and queries $f_{\Verifier}$. The
          designated honest prover queries $f_{\Prover}$.
\end{protodesc}

\begin{enumerate}
    \item The verifier draws independent samples $X_1,\ldots,X_m\sim D$
          and sends the ordered sample $S=(X_1,\ldots,X_m)$ to the
          prover.

    \item The honest prover sets $p_j:=f_{\Prover}(X_j)$, querying each
          distinct sampled point once, and sends $W:=\sum_{j=1}^m p_j$.

    \item The verifier rejects unless $W\in\mathcal{W}_K([m])$ and
          $\left|W/m-\sigma\right|\leq\epsc+a$.

    \item The parties perform $q_{\mathrm{b}}$ sequential descents of
          $\mathcal{T}_m$, each rooted at $W$, exactly as in Protocol
          $\Pi_{\mathrm{mis}}^{\mathrm{BB}}$. If a range or additivity
          check fails, the verifier rejects immediately. Otherwise, the
          descent reaches $\{J_t\}$ with claim $A_t$; the verifier
          queries $V_t:=f_{\Verifier}(X_{J_t})$ and sets
          $Z_t:=\mathbf{1}\bigl[|A_t-V_t|>\gamma\bigr]$.

    \item If no previous check failed, the verifier accepts if and only
          if
          \[
              \sum_{t=1}^{q_{\mathrm{b}}} Z_t
              \leq
              \left(\rho+\frac{\eta}{2}\right)q_{\mathrm{b}}.
          \]
\end{enumerate}
\end{protofig}

\begin{theorem}[Banded upper bound]
\label{thm:band-upper-bound}
Under the hypotheses of
Theorem~\ref{thm:corrupted-query-upper-bound} (in particular $\eta>0$),
Protocol $\Pi_{\mathrm{band}}^{\mathrm{BB}}$ is an
$(\epsc,\epsf)$-interactive proof with respect to
$\Delta_{\mathrm{mis}}$, with completeness and soundness errors at most
$\delta$. The verifier draws $m=O(\ln(1/\delta)/\eta^2)$ samples from
$D$ and makes
\[
    q_{\mathrm{b}}=O \left(\frac{(\rho+\eta)\ln(1/\delta)}{\eta^2}\right)
\]
queries to $f_{\Verifier}$. The honest prover draws no samples and makes
at most $\min\{m,|\mathcal{X}|\}$ distinct queries to $f_{\Prover}$. The
communication is $O(m\ell+q_{\mathrm{b}}\log m\log(mK))$ bits and the
number of rounds is $O(q_{\mathrm{b}}\log m)$.
\end{theorem}

\begin{proof}
First note that $\eta>0$ forces $\kappa<g/2\leq1/2$, hence $\rho<1/2$
and
$\gamma\leq\kappa/(1-\rho)<(1-2\rho)/(2(1-\rho))\leq1/2$;
in particular $1-\gamma\geq1/2$, so the band test carries useful information (as $\gamma \rightarrow 1$, the band test degenerates to a trivial test which always accepts) and
the factor $1/(1-\gamma)$ in Lemma~\ref{lem:banded-descent} is at most 2.

Let $V_\gamma:=\{x\in\mathcal{X}:
|f_{\Prover}(x)-f_{\Verifier}(x)|>\gamma\}$. By
Definition~\ref{def:fixed-score-mismatch}, $V_\gamma\subseteq B$, so
$D(V_\gamma)\leq\rho$. Define $\overline p$ and $\overline v$ as in the
proof of Theorem~\ref{thm:corrupted-query-upper-bound}, and let
$\overline b:=\frac1m|\{j\in[m]:X_j\in V_\gamma\}|$.

\paragraph{Completeness.}
Suppose $|\mu_{\Prover}-\sigma|\leq\epsc$ and the prover follows the
honest strategy. Hoeffding's inequality gives
\[
    \Pr[|\overline p-\mu_{\Prover}|>a]\leq2e^{-2ma^2}\leq\frac{\delta}{6},
    \qquad
    \Pr[\overline b>\rho+a]\leq e^{-2ma^2}\leq\frac{\delta}{12},
\]
the latter because $\overline b$ is an average of independent Bernoulli
variables with mean $D(V_\gamma)\leq\rho$. We now condition on the complementary events. Since $W/m=\overline p$, the root check passes.
Every consistency check passes, and by
Lemma~\ref{lem:banded-descent}, conditional on the sample the variables
$Z_1,\ldots,Z_{q_{\mathrm{b}}}$ are independent Bernoulli variables with
mean $\overline b\leq\rho+\eta/8$. Since the acceptance threshold
exceeds this mean by at least $3\eta/8$, Bernstein's inequality gives
\[
    \Pr\left[\sum_{t=1}^{q_{\mathrm{b}}}Z_t
        >\left(\rho+\frac{\eta}{2}\right)q_{\mathrm{b}}\right]
    \leq
    \exp\left(
        -\frac{q_{\mathrm{b}}(3\eta/8)^2}
              {2\bigl(\overline b(1-\overline b)+(3\eta/8)/3\bigr)}
    \right)
    \leq
    \exp\left(-\frac{9\,q_{\mathrm{b}}\,\eta^2}{128(\rho+\eta)}\right)
    \leq\frac{\delta}{12},
\]
using $\overline b(1-\overline b)+\eta/8\leq\rho+\eta/4\leq\rho+\eta$
and the definition of $q_{\mathrm{b}}$. The total rejection probability
is at most $\delta/6+\delta/12+\delta/12\leq\delta$.

\paragraph{Soundness.}
Suppose $|\mu_{\Prover}-\sigma|>\epsf$, and fix an arbitrary cheating
prover. Hoeffding's inequality gives
$\Pr[|\overline v-\mu_{\Verifier}|>a]\leq2e^{-2ma^2}\leq\delta/6$. Fix a
sample conditioned on the complementary (good) event ($|\overline v-\mu_{\Verifier}|\leq a$) and a root claim $W$ that passes the
root check. Lemma~\ref{lem:fixed-mismatch-mean} and the triangle
inequality imply
\[
    \left|\frac{W}{m}-\overline v\right|
    >\epsf-\kappa-a-(\epsc+a)
    =g-\kappa-\frac{\eta}{4}.
\]
Consider now a relaxed verifier that records $Z_t=1$ and continues whenever a consistency check fails (note that this only increases the acceptance probability). We consider the sampling of a continuation conditioned on the complete transcript before descent $t$ (the sample, the root claim $W$, and the transcript of the descents $1$ to $t-1$): Lemma~\ref{lem:banded-descent} gives
\[
    \Pr[Z_t=1\mid\text{prior transcript}]
    \geq
    \frac{g-\kappa-\eta/4-\gamma}{1-\gamma}
    =
    \frac{\rho(1-\gamma)+3\eta/4}{1-\gamma}
    \geq
    p_*:=\rho+\frac{3\eta}{4},
\]
where we used $g-\kappa-\gamma=\kappa+\eta-\gamma=\rho(1-\gamma)+\eta$.
Note $p_*\leq\rho+\eta\leq g\leq1$. A sequence of Bernoulli variables
whose conditional success probabilities are at least $p_*$
stochastically dominates
$\operatorname{Bin}(q_{\mathrm{b}},p_*)$, and the acceptance threshold
equals $(p_*-\eta/4)q_{\mathrm{b}}$. A Chernoff bound then yields
\[
    \Pr\left[\sum_{t=1}^{q_{\mathrm{b}}}Z_t
        \leq\left(\rho+\frac{\eta}{2}\right)q_{\mathrm{b}}\right]
    \leq
    \exp\left(-\frac{q_{\mathrm{b}}(\eta/4)^2}{2p_*}\right)
    \leq
    \exp\left(-\frac{q_{\mathrm{b}}\,\eta^2}{32(\rho+\eta)}\right)
    \leq\frac{\delta}{12}.
\]
After accounting for the sampling event, by a straightforward union bound, the cheating prover is
accepted with probability at most $\delta/6+\delta/12\leq\delta$. Lastly, the costs (depth and bits exchanged) are identical to
Theorem~\ref{thm:corrupted-query-upper-bound}, with $q_{\mathrm{b}}$ descents instead of $q$.
\end{proof}

%% file: sections/appendix/optimality_disagreement.tex
This appendix proves lower bounds (partially) matching
Theorem~\ref{thm:band-upper-bound} on three quantities: (range of supported parameters) the separation
$\eta=g-2\kappa$, which must be positive; (sample complexity) the number
$m=O(\ln(1/\delta)/\eta^2)$ of verifier samples; and (query complexity) the number
$q_{\mathrm{b}}=O((\rho+\eta)\ln(1/\delta)/\eta^2)$ of verifier queries.
We use the notation of Section~\ref{subsec:corrupted-query-model}. All
bounds hold for arbitrary protocols in the black-box model with bounded $(\rho,\gamma)$-mismatch. A recurring quantity is
\[
    \kappa^\star:=\rho+(1-2\rho)\gamma=\kappa-\rho\gamma,
    \qquad
    \eta^\star:=g-2\kappa^\star=\eta+2\rho\gamma.
\]
It coincides with $\kappa$ when $\rho=0$ or $\gamma=0$, in particular for
Boolean evaluations ($K=1$). For simplicity, we assume
that $\gamma$ is a multiple of $1/K$.

\subsection{Necessity of the separation condition}
\label{app:mismatch-threshold}

Protocol $\Pi_{\mathrm{band}}^{\mathrm{BB}}$ (like Protocol
$\Pi_{\mathrm{mis}}^{\mathrm{BB}}$) requires $\eta=g-2\kappa>0$. We show
that below the nearby threshold $2\kappa^\star$ the problem is
unsolvable (with any amount of resources) because $2\kappa^\star$ is the diameter of
the set of prover-side means consistent with a fixed verifier view.

\begin{lemma}[Reachable prover-side means]
\label{lem:reachable-means}
Let $\rho\leq1/2$ and $\gamma\leq1/2$. Fix $D$ and
$f_{\Verifier}\colon\mathcal{X}\to\mathcal{S}_K$. If
$(f^+,f_{\Verifier})$ and $(f^-,f_{\Verifier})$ both have bounded
$(\rho,\gamma)$-mismatch with respect to $D$, then
$\bigl|\mathbb{E}_D[f^+]-\mathbb{E}_D[f^-]\bigr|
\leq2\rho+2(1-2\rho)\gamma=2\kappa^\star$.
\end{lemma}

\begin{proof}
Let $B$ be the union of the two exceptional sets, so $D(B)\leq2\rho$.
Off $B$, the triangle inequality gives $|f^+-f^-|\leq2\gamma$, and
everywhere $|f^+-f^-|\leq1$. Hence
$\bigl|\mathbb{E}[f^+]-\mathbb{E}[f^-]\bigr|
\leq D(B)+(1-D(B))\,2\gamma\leq2\rho+(1-2\rho)\,2\gamma$, the last
step because $2\gamma\leq1$ makes the middle expression nondecreasing in
$D(B)$.
\end{proof}

\begin{theorem}[Impossibility below the threshold]
\label{thm:mismatch-impossibility}
Let $K\geq1$, $\rho\in[0,1/2]$, and $\gamma\in[0,1/2]$ with
$\gamma K\in\mathbb{N}$. Let $0\leq\epsc<\epsf\leq1$ satisfy
$\epsf<\epsc+2\kappa^\star$ and $\epsc+2\kappa^\star\leq1$. Then there
is no $(\epsc,\epsf)$-interactive proof with respect to
$\Delta_{\mathrm{mis}}$ for instances with bounded
$(\rho,\gamma)$-mismatch, regardless of the resources of both parties.
\end{theorem}

\begin{proof}
Suppose such a protocol exists, with designated honest prover
$\Prover_h$. Let $\mathcal{X}=\{b_0,b_1,u\}$ with $D(b_0)=D(b_1)=\rho$
and $D(u)=1-2\rho$, and define, listing values on $(b_0,b_1,u)$,
\[
    f_{\Verifier}:=(0,\;1,\;\gamma),
    \qquad
    f^+:=(1,\;1,\;2\gamma),
    \qquad
    f^-:=(0,\;0,\;0),
\]
all in $\mathcal{S}_K$ since $\gamma K\in\mathbb{N}$ and $2\gamma\leq1$.
The pair $(f^+,f_{\Verifier})$ has bounded $(\rho,\gamma)$-mismatch with
exceptional set $\{b_0\}$ (the discrepancies off it are $0$ on $b_1$ and
$\gamma$ on $u$), and $(f^-,f_{\Verifier})$ with exceptional set
$\{b_1\}$. The prover-side means are
$\mu^+=2\rho+(1-2\rho)\,2\gamma=2\kappa^\star$ and $\mu^-=0$. With
$\sigma:=\epsc+2\kappa^\star\in[0,1]$, the instance
$I^+:=(D,f^+,f_{\Verifier})$ is a completeness instance
($\Delta_{\mathrm{mis}}(I^+,\sigma)=\epsc$) and
$I^-:=(D,f^-,f_{\Verifier})$ is a soundness instance
($\Delta_{\mathrm{mis}}(I^-,\sigma)=\epsc+2\kappa^\star>\epsf$). On $I^-$, let the cheating prover run $\Prover_h$, answering its
$f_{\Prover}$-queries according to the explicit function $f^+$ (this is
legal since a cheating prover may know both instances). The two instances
share $D$ and $f_{\Verifier}$, which are all the verifier can access, so
the two interactions have identical transcript distributions and hence
equal acceptance probabilities. Completeness on $I^+$ makes this
probability at least $2/3$ and soundness on $I^-$ at most $1/3$, a
contradiction.
\end{proof}

Since Theorem~\ref{thm:band-upper-bound} assumes
$g>2\kappa=2\kappa^\star+2\rho\gamma$, its separation condition is exactly necessary when $\rho\gamma=0$, in particular for Boolean evaluations, and necessary up to the lower-order term $2\rho\gamma$ in general. We leave open the question of whether the remaining range $2\kappa^\star<g\leq2\kappa$ admits a sublinear protocol.

\subsection{Sample complexity lower bound}
\label{app:mismatch-sample-lb}

Our proof follows closely the proof of Theorem~\ref{thm:black-box-sample-lower-bound}: two
instances whose visible oracles differ only in the distribution, at
relative entropy $O({\eta^\star}^2)$ per sample (recall $\eta^\star =g-2\kappa^\star=\eta+2\rho\gamma$).

\begin{theorem}[Sample lower bound]
\label{thm:mismatch-sample-lb}
Let $K\geq1$, $\rho\in[0,1/2)$, and $\gamma\in[0,1/2)$ with
$\gamma K\in\mathbb{N}$, and suppose $\eta^\star>0$. Fix
$\delta\in(0,1/2)$, $\alpha\in(0,1/2)$, and a claim $\sigma\in[0,1]$ such
that
\[
    p_{\mathrm{Y}}:=\frac{\sigma+\epsc}{(1-2\rho)(1-2\gamma)}
    \qquad\text{and}\qquad
    p_{\mathrm{N}}:=p_{\mathrm{Y}}+\frac{2\eta^\star}{(1-2\rho)(1-2\gamma)}
\]
both lie in $[\alpha,1-\alpha]$. Every $(\epsc,\epsf)$-interactive proof
with respect to $\Delta_{\mathrm{mis}}$ for instances with bounded
$(\rho,\gamma)$-mismatch, with completeness and soundness errors at most
$\delta$, requires the verifier to draw at least
\[
    s\geq\frac{\alpha(1-\alpha)(1-2\rho)(1-2\gamma)^2(1-2\delta)^2}
              {2\,{\eta^\star}^2}
\]
samples from $D$ in the worst case, regardless of the number of rounds,
the communication, and the number of queries to $f_{\Verifier}$ and
$f_{\Prover}$. (Replacing every score $f$ by $1-f$ and $\sigma$ by
$1-\sigma$ covers claims near the other end of $[0,1]$.)
\end{theorem}

\begin{proof}
Let $\mathcal{X}=\{b_0,b_1,x_0,x_1\}$ and, for $p\in[0,1]$, let $D_p$
have $D_p(b_0)=D_p(b_1)=\rho$, $D_p(x_0)=(1-2\rho)(1-p)$, and
$D_p(x_1)=(1-2\rho)p$. Define the score functions
\[
    f_{\Verifier}:=(0,\;1,\;\gamma,\;1-\gamma),
    \qquad
    f^{\mathrm{Y}}:=(0,\;0,\;0,\;1-2\gamma),
    \qquad
    f^{\mathrm{N}}:=(1,\;1,\;2\gamma,\;1),
\]
all in $\mathcal{S}_K$, where we index the values on $(b_0,b_1,x_0,x_1)$. With respect to every $D_p$, the pair
$(f^{\mathrm{Y}},f_{\Verifier})$ has bounded $(\rho,\gamma)$-mismatch
with exceptional set $\{b_1\}$, and $(f^{\mathrm{N}},f_{\Verifier})$
with exceptional set $\{b_0\}$. Writing
$\mu_{\Verifier}(p):=\rho+(1-2\rho)\bigl(\gamma+p(1-2\gamma)\bigr)$ for
the verifier-side mean, each function shifts it by $\rho$ on its
exceptional set and by $\gamma$ on the remaining mass $1-2\rho$, so
$\mathbb{E}_{D_p}[f^{\mathrm{Y}}]=\mu_{\Verifier}(p)-\kappa^\star$ and
$\mathbb{E}_{D_p}[f^{\mathrm{N}}]=\mu_{\Verifier}(p)+\kappa^\star$.

Let $I_{\mathrm{Y}}:=(D_{p_{\mathrm{Y}}},f^{\mathrm{Y}},f_{\Verifier})$
and $I_{\mathrm{N}}:=(D_{p_{\mathrm{N}}},f^{\mathrm{N}},f_{\Verifier})$.
The definition of $p_{\mathrm{Y}}$ gives
$\mathbb{E}_{D_{p_{\mathrm{Y}}}}[f^{\mathrm{Y}}]=\sigma+\epsc$, so
$I_{\mathrm{Y}}$ is a valid completeness instance; the definition of
$p_{\mathrm{N}}$ and $g=2\kappa^\star+\eta^\star$ give
$\mathbb{E}_{D_{p_{\mathrm{N}}}}[f^{\mathrm{N}}]
=\sigma+\epsc+2\eta^\star+2\kappa^\star=\sigma+\epsf+\eta^\star$, so
$I_{\mathrm{N}}$ is a valid soundness instance. The two instances share
$f_{\Verifier}$ and differ only in the distribution.

On $I_{\mathrm{N}}$, let the cheating prover run the designated honest
prover with fresh coins, answering its $f_{\Prover}$-queries according
to $f^{\mathrm{Y}}$; since the honest prover has no sampling access,
this reproduces its behavior on $I_{\mathrm{Y}}$ exactly. The coupling
and data-processing argument of
Theorem~\ref{thm:black-box-sample-lower-bound} now applies verbatim: if
the verifier draws at most $s$ samples, its views in the two executions
are post-processings of $D_{p_{\mathrm{Y}}}^{\otimes s}$ and
$D_{p_{\mathrm{N}}}^{\otimes s}$, so their total variation distance,
which is at least $1-2\delta$, is at most
$\sqrt{s\,\KL(D_{p_{\mathrm{Y}}}\|D_{p_{\mathrm{N}}})/2}$. Since the two
distributions agree on $b_0$ and $b_1$,
\[
    \KL\bigl(D_{p_{\mathrm{Y}}}\,\big\|\,D_{p_{\mathrm{N}}}\bigr)
    =(1-2\rho)\,\KL\bigl(\Ber(p_{\mathrm{Y}})\,\big\|\,\Ber(p_{\mathrm{N}})\bigr)
    \leq(1-2\rho)\frac{(p_{\mathrm{N}}-p_{\mathrm{Y}})^2}
                      {p_{\mathrm{N}}(1-p_{\mathrm{N}})}
    \leq\frac{4{\eta^\star}^2}{(1-2\rho)(1-2\gamma)^2\alpha(1-\alpha)},
\]
and rearranging gives the claim.
\end{proof}

For fixed $\delta$ and $\alpha$ this is $s=\Omega(1/{\eta^\star}^2)$,
against $m=O(\ln(1/\delta)/\eta^2)$ in
Theorem~\ref{thm:band-upper-bound}: the sample complexity is
$\Theta(1/\eta^2)$ when $\rho\gamma=0$, and matches up to the
replacement of $\eta$ by $\eta^\star$ in general. As in
Remark~\ref{rem:black-box-error-dependence}, the Bretagnolle--Huber
inequality shows that the $\ln(1/\delta)$ dependence is necessary as
well. For $\rho=\gamma=0$ the statement reduces to
Theorem~\ref{thm:black-box-sample-lower-bound}.

\subsection{Verifier query lower bounds}
\label{app:mismatch-query-lb}

There are two orthogonal reasons that force verifier queries in
Theorem~\ref{thm:band-upper-bound}. First, the
query lower bounds of Section~\ref{sec:hamming_weight} apply, since the evaluator-disagreement model contains the exact model as a special case (set $f_{\Prover}=f_{\Verifier}$). Second, when the exceptional mass $\rho$ is comparable to $g$, the verifier's task boils down to estimating a band-violation rate of order $\rho$ to additive precision of order $\eta$, which costs $\Omega(\rho/\eta^2)$ queries.

\begin{proposition}[Transfer from the exact model]
\label{prop:mismatch-query-reduction}
Let $K\geq1$, $\rho,\gamma\in[0,1]$, $0\leq\epsc<\epsf\leq1/4$, and
$1/g^2\leq n$. Every $(\epsc,\epsf)$-interactive proof with respect to
$\Delta_{\mathrm{mis}}$ for instances with bounded
$(\rho,\gamma)$-mismatch over domains of size $n$ satisfies
$Q_\Verifier=\Omega(1/g)$ and $Q_\Verifier+Q_\Prover=\Omega(1/g^2)$,
where $Q_\Verifier$ counts queries to $f_{\Verifier}$ and $Q_\Prover$
the honest prover's queries to $f_{\Prover}$.
\end{proposition}

We omit the straightforward proof. The second bound shows that when $\rho$ is a constant fraction of $g$,
so that $\eta\ll g$, the $\rho/\eta^2$ term in
Theorem~\ref{thm:band-upper-bound} is unavoidable. The bound uses Boolean instances, hence it is valid for every $\gamma$ and $K$. We write
$\eta_0:=g-2\rho$, which equals $\eta$ when $\gamma=0$ and satisfies
$\eta\leq\eta^\star\leq\eta_0$ in general.

\begin{theorem}[Disagreement-estimation lower bound]
\label{thm:mismatch-disagreement-lb}
There is a universal constant $c>0$ such that the following holds. Let
$K\geq1$, $\gamma\in[0,1]$, $0<\rho<g/2$, $0\leq\epsc<\epsf\leq1/16$,
$\sigma\in[1/2,3/4]$, and $n\geq1/\eta_0^2$. Every
$(\epsc,\epsf)$-interactive proof with respect to
$\Delta_{\mathrm{mis}}$ for instances with bounded
$(\rho,\gamma)$-mismatch over domains of size $n$, with completeness and
soundness errors at most $1/3$, satisfies
$Q_\Verifier\geq c\,\rho/\eta_0^2$.
\end{theorem}

\begin{proof}
Set $\theta:=\sigma-\rho-\epsc$, $\tau:=\theta+\rho=\sigma-\epsc$, and
$\theta':=\theta-2\eta_0$. The hypotheses give $\rho<1/32$ and
$\epsc,\eta_0\leq1/16$, hence
\begin{equation}
\label{eq:disagreement-lb-ranges}
    \tfrac{13}{32}\leq\theta\leq\tau\leq\tfrac{25}{32},
    \qquad
    \theta'\geq\tfrac{9}{32},
    \qquad
    \rho+2\eta_0\leq\tfrac{5}{32}.
\end{equation}
Let $D$ be uniform over $[n]$ (note that the verifier's samples are then uniform indices that carry no information about the instance).

\paragraph{Completeness ensemble.} Draw $t\in\{0,1\}^n$ uniformly among the
strings of weight $\tau n$ and a uniformly random
$S\subseteq\operatorname{supp}(t)$ with $|S|=\rho n$; set
$f_{\Prover}:=t$ and $f_{\Verifier}:=t\wedge\lnot\mathbf{1}_S$. The two
functions differ exactly on $S$, with $D(S)=\rho$, so the instance is
valid, and $\mu_{\Prover}=\tau=\sigma-\epsc$, so it is a completeness
instance. The prover is the designated honest prover with oracle $t$.

\paragraph{Soundness ensemble.} Draw $f_{\Verifier}$ uniformly among the
strings of weight $\theta'n$ and a uniformly random
$S'\subseteq\operatorname{supp}(f_{\Verifier})$ with $|S'|=\rho n$; set
$f_{\Prover}:=f_{\Verifier}\wedge\lnot\mathbf{1}_{S'}$. The instance is
valid, and
$\mu_{\Prover}=\theta'-\rho=\sigma-\epsc-2\eta_0-2\rho
=\sigma-\epsf-\eta_0$, so it is a soundness instance. The cheating
prover $\Prover^*$ draws a uniformly random
$\widehat S\subseteq[n]\setminus\operatorname{supp}(f_{\Verifier})$ with
$|\widehat S|=(\rho+2\eta_0)n$, which is feasible
by~\eqref{eq:disagreement-lb-ranges}, sets
$\widehat g:=f_{\Verifier}\vee\mathbf{1}_{\widehat S}$, and runs the
designated honest prover algorithm with oracle $\widehat g$, ignoring
$f_{\Prover}$.

\paragraph{Common representation.} In both ensembles, the string given to the
prover algorithm ($t$, respectively $\widehat g$) is uniformly
distributed among the strings of weight $\tau n$, since
$\theta'+\rho+2\eta_0=\tau$, and, conditionally on it, $f_{\Verifier}$ is
obtained by choosing a uniformly random subset $R$ of its support and
setting those coordinates to $0$, with $|R|=\rho n$ in the completeness
ensemble and $|R|=(\rho+2\eta_0)n$ in the soundness ensemble. Indeed,
both constructions are invariant under permutations of the positions,
and both produce a pair $(w,w')$ with $w'\leq w$ coordinatewise and
prescribed weights; since any two such pairs are related by a
permutation, each construction yields the uniform distribution over
these pairs, and the two descriptions agree. The true prover-side
function of the soundness instance affects neither the prover's messages
nor the query answers. Hence, in either world, the verifier's view is
generated by drawing $w$ uniformly among the strings of weight $\tau n$
and $R$ as above, running the interaction with the prover algorithm on
oracle $w$, and answering each query to $f_{\Verifier}$ at $i$ by
$w_i\cdot\mathbf{1}[i\notin R]$.

\paragraph{Divergence.} Reveal $w$ and all random tapes to the distinguisher;
their joint law is identical in the two worlds. The transcript is then a
deterministic function of the successive answers to the verifier's
$f_{\Verifier}$-queries, and only the first query at each position of
$\operatorname{supp}(w)$ carries information, since queries outside the
support return $0$ in both worlds. Suppose $Q_\Verifier<\rho n/2$. For
an informative query made after $j$ earlier ones of which $r$ returned
$0$, the probability that the answer is $0$ is
$u=(\rho n-r)/(\tau n-j)$ in the completeness world and
$v=((\rho+2\eta_0)n-r)/(\tau n-j)$ in the soundness world, conditionally
on any past and however the position was chosen. Using
$j,r<\rho n/2\leq n/64$ and~\eqref{eq:disagreement-lb-ranges},
\[
    \tau n-j\geq\tfrac38n,
    \qquad
    \tfrac{\rho}{2}\leq u,
    \qquad
    \tfrac{\rho}{2}\leq v\leq\tfrac83(\rho+2\eta_0)\leq\tfrac12,
    \qquad
    v-u=\frac{2\eta_0n}{\tau n-j}\leq\tfrac{16}{3}\eta_0,
\]
so $v(1-v)\geq\rho/4$ and
$\KL(\Ber(u)\|\Ber(v))\leq(u-v)^2/(v(1-v))\leq114\,\eta_0^2/\rho$. By
the chain rule for relative entropy, the divergence between the two
complete views is at most $114\,Q_\Verifier\,\eta_0^2/\rho$.

\paragraph{Conclusion.} Every completeness instance is accepted with
probability at least $2/3$ and every soundness instance, against
$\Prover^*$, with probability at most $1/3$, so the two views are at
total variation distance at least $1/3$, and Pinsker's inequality gives
$114\,Q_\Verifier\,\eta_0^2/\rho\geq2/9$, that is,
$Q_\Verifier\geq\rho/(513\,\eta_0^2)$. If instead
$Q_\Verifier\geq\rho n/2$, then $n\geq1/\eta_0^2$ gives
$Q_\Verifier\geq\rho/(2\eta_0^2)$. The theorem holds with $c=1/513$.
\end{proof}

\subsection{Bottom line}
\label{app:mismatch-bottom-line}

\begin{corollary}[Optimality of Protocol $\Pi_{\mathrm{band}}^{\mathrm{BB}}$]
\label{cor:query-landscape}
Fix a constant error probability and restrict to the interior parameter
regimes of Theorems~\ref{thm:mismatch-sample-lb}
and~\ref{thm:mismatch-disagreement-lb} and
Proposition~\ref{prop:mismatch-query-reduction}. Write
$\eta_0=g-2\rho=\eta+2(1-\rho)\gamma$ and
$\eta^\star=\eta+2\rho\gamma$. For Protocol
$\Pi_{\mathrm{band}}^{\mathrm{BB}}$ of Theorem~\ref{thm:band-upper-bound}:
\begin{enumerate}
    \item the separation condition $\eta>0$ is necessary up to the
          additive term $2\rho\gamma$, and exactly necessary when
          $\rho\gamma=0$;
    \item the verifier sample complexity $\Theta(1/\eta^2)$ is optimal
          up to constants when $\rho\gamma=0$, and up to the replacement
          of $\eta$ by $\eta^\star$ in general;
    \item if $\gamma=0$, in particular for Boolean evaluations, the
          verifier query complexity
          $\Theta\bigl((\rho+\eta)/\eta^2\bigr)=\Theta(g/\eta^2)$ is
          optimal up to constants, and so is that of Protocol
          $\Pi_{\mathrm{mis}}^{\mathrm{BB}}$, which coincides with it in
          this case;
    \item in general,
          $\Omega(\rho/\eta_0^2+1/g)\leq Q_\Verifier
          \leq O\bigl(\rho/\eta^2+1/\eta\bigr)$.
\end{enumerate}
Thus, for Boolean evaluations, the protocol is optimal in all three
quantities. For $\gamma>0$, the threshold and the sample complexity are
optimal up to the lower-order term $2\rho\gamma$, while the query lower
bounds may be loose by factors polynomial in $\gamma/\eta$; we do not
pursue this here.
\end{corollary}

Item 3 stems from the fact that $\gamma=0$ gives $g=2\rho+\eta$, hence
\[
    \frac{\rho+\eta}{\eta^2}
    \leq\frac{g}{\eta^2}
    =\frac{2\rho}{\eta^2}+\frac{1}{\eta}
    \leq4\max\left\{\frac1g,\frac{\rho}{\eta^2}\right\},
\]
where the last inequality uses that either $\rho\leq\eta/2$, so that
$g\leq2\eta$ and $1/\eta\leq2/g$, or $\rho>\eta/2$, so that
$1/\eta<2\rho/\eta^2$.